\documentclass[12pt,authoryear]
{article}
\usepackage[margin=2cm]{geometry}

\usepackage{authblk}
\usepackage{amsmath}
\usepackage{amssymb}
\usepackage{amsfonts}
\usepackage{cancel}
\usepackage{amsthm}
\usepackage{mathtools}
\usepackage{accents}
\usepackage{comment}
\allowdisplaybreaks

\usepackage[T1]{fontenc}
\usepackage[mathscr]{eucal}
\usepackage{dsfont}

\usepackage{fullpage}

\usepackage{setspace}

\usepackage{acronym}
\usepackage{enumitem}
\usepackage{latexsym}
\usepackage{longtable}
\usepackage{wasysym}
\usepackage{xspace}

\usepackage[dvipsnames,svgnames]{xcolor}
\colorlet{MyBlue}{DodgerBlue!60!Black}
\colorlet{MyGreen}{DarkGreen!85!Black}

\usepackage[font=small,labelfont=bf]{caption}
\usepackage{tikz}
\usetikzlibrary{positioning, arrows.meta}
\usetikzlibrary{calc,patterns}
\usepackage{graphicx}
\usepackage{subcaption}
\graphicspath{{journal/PlotSameMu/}}
\usepackage[authoryear,comma]{natbib}

\usepackage{hyperref}
\hypersetup{
colorlinks=true,
linktocpage=true,
pdfstartview=FitH,
breaklinks=true,
pdfpagemode=UseNone,
pageanchor=true,
pdfpagemode=UseOutlines,
plainpages=false,
bookmarksnumbered,
bookmarksopen=false,
bookmarksopenlevel=1,
hypertexnames=true,
pdfhighlight=/O,
urlcolor=MyBlue!60!black,
linkcolor=MyBlue!70!black,
citecolor=DarkGreen!70!black, 
pdftitle={},
pdfauthor={},
pdfsubject={},
pdfkeywords={},
pdfcreator={pdfLaTeX},
pdfproducer={LaTeX with hyperref}
}

\numberwithin{equation}{section}  
\usepackage[sort&compress,capitalize,nameinlink]{cleveref}
\crefname{app}{Appendix}{Appendices}

\crefrangeformat{equation}{\upshape(#3#1#4)\textendash(#5#2#6)}

\usepackage[textwidth=30mm]{todonotes}

\usepackage{soul}
\setstcolor{red}
\sethlcolor{SkyBlue}

\newcommand{\debug}[1]{#1}

\usepackage{verbatim}

\usepackage[linesnumbered,ruled,vlined]{algorithm2e}

\SetCommentSty{mycommfont} 
\SetAlgoSkip{bigskip} 

\theoremstyle{plain}
\newtheorem{theorem}{Theorem}

\AddToHook{env/corollary/begin}{\crefalias{theorem}{corollary}}
\newtheorem{lemma}[theorem]{Lemma}
\AddToHook{env/lemma/begin}{\crefalias{theorem}{lemma}}
\newtheorem{proposition}[theorem]{Proposition}
\AddToHook{env/proposition/begin}{\crefalias{theorem}{proposition}} 

\newtheorem{claim}[theorem]{Claim}
\AddToHook{env/claim/begin}{\crefalias{theorem}{claim}}

\theoremstyle{definition}
\newtheorem{definition}[theorem]{Definition}
\AddToHook{env/definition/begin}{\crefalias{theorem}{definition}}

\AddToHook{env/remark/begin}{\crefalias{theorem}{remark}}
\newtheorem{example}[theorem]{Example}
\AddToHook{env/example/begin}{\crefalias{theorem}{example}}

\numberwithin{theorem}{section}

\DeclarePairedDelimiter{\braces}{\{}{\}}
\DeclarePairedDelimiter{\bracks}{[}{]}
\DeclarePairedDelimiter{\parens}{(}{)}
\DeclarePairedDelimiter{\abs}{\lvert}{\rvert}

\DeclarePairedDelimiterX{\braket}[2]{\langle}{\rangle}{#1,#2}

\DeclarePairedDelimiterX{\inner}[2]{\langle}{\rangle}{#1,#2}
\DeclarePairedDelimiterX{\setdef}[2]{\{}{\}}{#1:#2}

\DeclarePairedDelimiterXPP{\probof}[1]{\Prob}{(}{)}{}{%

#1}

\DeclarePairedDelimiterXPP{\exof}[1]{\Expect}{[}{]}{}{%

#1}

\newcommand{\lipschitz}{\debug L}
\newcommand{\smooth}{\debug C}

\newcommand{\proba}{\debug q}
\newcommand{\dirac}{\debug \delta}

\newcommand{\stateth}{\debug \theta}

\newcommand{\states}{\debug \Theta}

\newcommand{\prior}{\debug p}

\newcommand{\simplex}{\debug \Delta}

\DeclareMathOperator{\Expect}{\mathsf{\debug{E}}}
\DeclareMathOperator{\Prob}{\mathsf{\debug{P}}}

\newcommand{\act}{\debug a}
\newcommand{\actalt}{\debug b}

\newcommand{\actions}{\mathcal{\debug A}}

\newcommand{\util}{\debug u}
\newcommand{\utilprof}{\boldsymbol{\util}}
\newcommand{\Util}{\debug U}

\newcommand{\cost}{\debug c}

\newcommand{\outcome}{\debug \mu}

\newcommand{\game}{\debug \Gamma}

\newcommand{\infostr}{\mathcal{\debug I}}

\newcommand{\potential}{\debug \Phi}

\newcommand{\type}{\debug \tau}
\newcommand{\typeprof}{\boldsymbol{\type}}
\newcommand{\types}{\mathcal{\debug T}}

\newcommand{\eq}[1]{#1^{\ast}}

\DeclareMathOperator{\WE}{\mathsf{\debug{WE}}}

\DeclareMathOperator{\BWespE}{\mathsf{\debug{BWE_{\varepsilon}}}}
\DeclareMathOperator{\BWalpE}{\mathsf{\debug{BWE_{\alpha}}}}

\DeclareMathOperator{\socialcost}{\mathsf{\debug{U}}}

\newcommand{\edges}{\mathcal{\debug E}}

\newcommand{\edge}{\debug e}

\newcommand{\flow}{\debug y}
\newcommand{\Flow}{\debug Y}
\newcommand{\flows}{\mathcal{\Flow}}

\newcommand{\flowprof}{\boldsymbol{\flow}}

\newcommand{\flowz}{\debug z}

\newcommand{\flowprofz}{\boldsymbol{\flowz}}
\newcommand{\floww}{\debug w}
\newcommand{\flowprofw}{\boldsymbol{\floww}}
\newcommand{\suppflows}{\eq{\flows}}

\newcommand{\load}{\debug x}

\newcommand{\reals}{\mathbb{R}}
\newcommand{\naturals}{\mathbb{N}}

\newcommand{\argdot}{\,\cdot\,}
\newcommand{\diff}{\ \textup{d}}

\newcommand{\run}{\debug n}

\newcommand{\ie}{i.e.,\ }
\newcommand{\eg}{e.g.,\ }

\DeclareMathOperator{\ind}{\mathds{\debug 1}}

\DeclareMathOperator{\prob}{\mathsf{\debug P}}

\newcommand{\gfunc}{\debug g}

\newcommand{\const}{\debug C}
\newcommand{\quantity}{\debug Q}
\newcommand{\invdem}{\debug \rho}

\newcommand{\pop}{\debug k}
\newcommand{\popalt}{\debug j}
\newcommand{\pops}{\mathcal{\debug K}}
\newcommand{\sizepop}{\debug \gamma}
\newcommand{\sizepopprof}{\boldsymbol{\sizepop}}
\newcommand{\prtype}{\debug \pi}
\newcommand{\setA}{\debug A}

\newcommand{\napproxim}{\debug N}
\newcommand{\error}{\debug \eta}

\newcommand{\modcont}{\debug \omega}

\newcommand{\flowak}{\flow_{\act}^{\pop}}
\newcommand{\flowaktype}{\flowak(\type^{\pop})}
\newcommand{\flowktypeprof}{\flowprof^{\pop}(\type^{\pop})}

\newcommand{\nat}{\debug K}

\newcommand{\descost}{\debug \psi}

\DeclareMathOperator{\Qmeas}{\mathsf{\debug Q}}

\newcommand{\acdef}[1]{\acfi{#1} \acused{#1}}

\newacro{AG}{anonymous game}
\newacro{ACG}{atomic congestion game}
\newacro{ACGSD}{atomic congestion game with stochastic demand}
\newacro{ANG}{nonatomic game}

\newacro{IIANG}{incomplete information anonymous nonatomic game}
\newacro{IIAG}{incomplete information anonymous game}

\newacro{BANG}{Bayesian  nonatomic game}

\newacro{PoA}{price of anarchy}
\newacro{PoS}{price of stability}
\newacro{TC}{total cost}
\newacro{TEC}{total expected cost}
\newacro{SO}{social optimum}
\newacro{SOC}{socially optimum cost}
\newacro{DC}{designer cost}

\newacro{NE}{Nash equilibrium}
\newacroplural{NE}[NE]{Nash equilibria}
\newacro{CE}{correlated equilibrium}
\newacroplural{CE}[CE]{correlated equilibria}
\newacro{CCE}{coarse correlated equilibrium}
\newacroplural{CCE}[CCE]{coarse correlated equilibria}
\newacro{PNE}{pure Nash equilibrium}
\newacroplural{PNE}[PNE]{pure Nash equilibria}

\newacro{BNE}{Bayes Nash equilibrium}
\newacroplural{BNE}[BNE]{Bayes Nash equilibria}
\newacro{BCE}{Bayes correlated equilibrium}
\newacroplural{BCE}[BCE]{Bayes correlated equilibria}
\newacro{BCCE}{Bayes coarse correlated equilibrium}
\newacroplural{BCCE}[BCCE]{Bayes coarse correlated equilibria}

\newacro{ENE}[$\varepsilon$-NE]{$\varepsilon$-Nash equilibrium}
\newacroplural{ENE}[$\varepsilon$-NE]{$\varepsilon$-Nash equilibria}
\newacro{ECE}[$\varepsilon$-CE]{$\varepsilon$-correlated equilibrium}
\newacroplural{ECE}[$\varepsilon$-CE]{$\varepsilon$-correlated equilibria}
\newacro{CCE}{coarse correlated equilibrium}
\newacroplural{CCE}[CCE]{coarse correlated equilibria}
\newacro{PNE}{pure Nash equilibrium}
\newacroplural{PNE}[PNE]{pure Nash equilibria}

\newacro{EBNE}[$\varepsilon$-BNE]{$\varepsilon$-Bayes Nash equilibrium}
\newacroplural{EBNE}[$\varepsilon$-BNE]{$\varepsilon$-Bayes Nash equilibria}
\newacro{EBCE}[$\varepsilon$-BCE]{$\varepsilon$-Bayes correlated equilibrium}
\newacroplural{EBCE}[$\varepsilon$-BCE]{$\varepsilon$-Bayes correlated equilibria}
\newacro{BCCE}{Bayes coarse correlated equilibrium}
\newacroplural{BCCE}[BCCE]{Bayes coarse correlated equilibria}

\newacro{CCWE}{coarse correlated Wardrop equilibrium}
\newacroplural{CCWE}[CCWE]{coarse correlated Wardrop equilibria}

\newacro{CBCWE}{coarse Bayes correlated Wardrop equilibrium}
\newacroplural{CBCWE}[CBCWE]{coarse Bayes correlated Wardrop equilibria}

\newacro{SBCWE}{simple Bayes correlated Wardrop equilibrium}
\newacroplural{SBCWE}[SBCWE]{simple Bayes correlated Wardrop equilibria}

\newacro{WE}{Wardrop equilibrium}
\newacroplural{WE}[WE]{Wardrop equilibria}

\newacro{BWE}{Bayesian Wardrop equilibrium}
\newacroplural{BWE}[BWE]{Bayesian Wardrop equilibria}

\newacro{BWepsE}[BW$\varepsilon$E]{Bayesian Wardrop $\varepsilon$-equilibrium}
\newacroplural{BWepsE}[BW$\varepsilon$E]{Bayesian Wardrop $\varepsilon$-equilibria}
\newacro{BWalpE}[BW$\alpha$E]{Bayesian Wardrop $\alpha$-equilibrium}
\newacroplural{BWalpE}[BW$\alpha$E]{Bayesian Wardrop $\alpha$-equilibria}

\newacro{CWE}{correlated Wardrop equilibrium}
\newacroplural{CWE}[CWE]{correlated Wardrop equilibria}

\newacro{BDWE}[BDWE]{Bayes deterministic Wardrop equilibrium}
\newacroplural{BDWE}[BDWE]{Bayes deterministic Wardrop equilibria}

\newacro{BCWE}{Bayes correlated Wardrop equilibrium}
\newacroplural{BCWE}[BCWE]{Bayes correlated Wardrop equilibria}

\newacro{KKT}{Karush\textendash Kuhn\textendash Tucker}
\newacro{OD}[OD]{origin-destination}
\newacro{BPR}{Bureau of Public Roads}

\newacro{SP}{series-parallel}

\begin{document}

\title{Information Design and Full Implementation in Nonatomic Games  \thanks{
Marco Scarsini is a member of GNAMPA-INdAM.
His work was partially supported by the  MIUR PRIN 2022EKNE5K ``Learning in markets and society,'' and the European Union-Next Generation EU Grant P2022XT8C8, component M4C2,
investment 1.1.
Frederic Koessler and Tristan Tomala gratefully acknowledge the support of the HEC foundation and the French National Research Agency under grant ANR SCOCOS (ANR-25-CE26-0626).} 
}
\author{Frederic Koessler%
\footnote{HEC Paris and GREGHEC-CNRS, 78351 Jouy-en-Josas, France. \textsf{koessler@hec.fr}}
\and
Marco Scarsini%
\footnote{Luiss University, 00197 Rome, Italy. \textsf{marco.scarsini@luiss.it}}
\and
Tristan Tomala%
\footnote{HEC Paris and GREGHEC, 78351 Jouy-en-Josas, France. \textsf{tomala@hec.fr}}
}

\maketitle

\begin{abstract}
 
This paper studies the implementation of Bayes correlated equilibria in symmetric Bayesian games with nonatomic players, using direct information structures and obedient strategies. 
The main results demonstrate full implementation in a class of games with negative payoff externalities, such as congestion and Cournot games. 
Specifically, if the game admits a strictly concave potential in every state, then for every Bayes correlated equilibrium outcome with finite support and rational action distributions, there exists a direct information structure that implements this outcome under all equilibria. 
When the potential is weakly concave, we show that all equilibria implement the same expected total payoff. 
Additionally, all Bayes correlated equilibria, including those with infinite support or irrational action distributions, are  approximately implemented.

\medskip \noindent \textsc{Keywords}: Bayes correlated equilibrium, 
congestion games, 
aggregative games, 
information design, 
nonatomic games, 
population games, 
potential games.

\end{abstract}

%
%

\clearpage

%
%
%
%

\section{Introduction}
\label{se:introduction}
The fundamental equilibrium concepts  used in information design for incomplete-information games with  finitely many players, are \acl{BNE} and \acl{BCE}.
The use of these concepts in games with a large number of players is often cumbersome and computationally intractable. 
Yet, large  games with symmetries can be approximated by nonatomic games, whose natural equilibrium concept in the case of complete information---\acl{WE}---is based on distributions over actions (called flows), rather than on the behavior of single players. 
This technique was recently extended to incomplete-information nonatomic games, and the concepts of \ac{BWE} and \ac{BCWE} were introduced and analyzed \citep[see][]{KoeScaTom:MOR2025}.

In this paper, we study partial and full implementation through information design in symmetric nonatomic games with incomplete information. 
In these, there is a payoff-relevant state, to which information will be correlated by a perfectly informed designer whose goal is---for instance---to maximize total welfare. 
In this context, the concept of \ac{BCWE} characterizes the set of state-dependent probability distributions over flows that can be induced in an incentive-compatible way by some information structure.
Furthermore, regardless of the designer's objective, a designer-optimal \ac{BCWE} can be formulated as a solution to an optimization program with a bounded number of variables and constraints, determined by the number of actions and states.

We show that approaching information design through the lens of nonatomic games provides important insights, particularly for achieving full implementation through direct recommendation systems. 
There is full implementation when \emph{all} equilibria implement the same outcome for the chosen information structure.
This property is crucial because it ensures robustness of the information design problem to adversarial equilibrium selection.
When only \emph{some} equilibria induce the desired outcome, we talk about \emph{partial} implementation.

Our basic game is a nonatomic game with symmetric and incomplete information.
Our analysis extends the class of population games studied by \citet{San:JET2001,San:MIT2010} to the Bayesian setting.%
\footnote{\citet{San:JET2001,San:MIT2010} allowed for multiple populations with different preferences and sets of feasible actions. 
For notational simplicity, in this paper we assume that the basic game involves only one population, so players are ex-ante symmetric. 
At the cost of cumbersome notation,  our results can be generalized to multiple populations.}
We extend this basic game to asymmetric information by considering a family of simple information structures \citep[as defined in][]{WuAmiOzd:OR2021}: 
The set of players is partitioned into a finite number of populations; within any given population, all players receive the same signal. 
Such an information structure is a \emph{direct information structure} if the set of possible signals is identical to the set of possible actions. 
A \ac{BWE} of the game extended with an information structure is said to be \emph{obedient} if all players in all populations follow their recommended action.

We first show that every \ac{BCWE} with finite support and rational flows can be partially implemented as an obedient \ac{BWE} with an appropriate direct information structure (\cref{thm:partialBI}).
Hence, a version of the revelation principle \citep{Mye:JME1982} applies to symmetric nonatomic  Bayesian games. 
Next, we turn to full implementation, which ensures that every \ac{BWE} outcome of the extended game leads to the desired \ac{BCWE}. 
We show that if the basic game admits a strictly concave potential in each state, then the \ac{BWE} outcome is unique for every information structure (\cref{thm:fullBI}). 
Consequently, information design becomes very effective in nonatomic games with a strictly concave potential because every \ac{BCWE} with finite support and rational flows can be fully implemented, \ie there exists a direct information structure such that the \ac{BCWE} is implemented, regardless of the \ac{BWE} selection.

If in each state the basic game admits a potential that is just weakly concave, full implementation of a given outcome may not be achieved, yet we get \emph{full payoff implementation}. 
We show that for every \ac{BCWE} with finite support and rational flows, there exists a direct information structure such that  expected total payoff is the same as the \ac{BCWE} expected total payoff regardless of the \ac{BWE} (\cref{thm:fullBI2}). 
In particular, the expected payoff of the total payoff-maximizing \ac{BCWE} can be fully implemented.

Similarly, approximate partial and full implementation results apply for all \ac{BCWE}, including those with infinite support and irrational flows. 
Specifically, for every \ac{BCWE} and for every $\varepsilon > 0$, an outcome that is arbitrarily close to that \ac{BCWE} can be implemented (partially or fully) as a \acdef{BWepsE}.

Our full (payoff) implementation results apply, in particular, to congestion games with decreasing payoff functions and to Cournot games.
A congestion game has a strictly concave potential if the set of actions is equal to the set of resources and the payoff of using a resource is strictly decreasing in the number of users. 
When actions are general feasible subsets of resources and payoffs are added across resources, congestion games admit a concave potential, provided that the payoff of using each resource is weakly decreasing in the number of users.  
A Cournot game admits a concave potential if inverse demand is strictly decreasing, and the potential is strictly concave  when possible quantities for each firm are binary (\eg entry and no entry).

%
%
%
%

\subsection{Related Literature}
\label{suse:literature-review}

The notion of Wardrop equilibrium originated from the work of \citet{War:PICE1952} on traffic flow theory, where he formulated principles describing user-optimal and system-optimal routing in transportation networks. 
\citet{BecMcGWin:Yale1956} proved that a Wardrop equilibrium is the solution of a convex optimization problem.
This idea was framed in the language of potential games  \citep[see, among others,][]{San:JET2001,Rou:AGT2007,RouTar:AGT2007}, generalizing the case of a finite number of players \citep{Ros:IJGT1973,MonSha:GEB1996} to the nonatomic case.
A general framework for nonatomic potential games, extending beyond routing games, was studied by \citet{San:JET2001,HofSan:JET2009,San:MIT2010}. 
The class of aggregative games---subclass of potential games---was implicitly introduced by \citet{DubMasShu:JET1980} and formalized by \citet{Cor:MSS1994}.
\citet{Lah:DGA2017} studied this class in games with nonatomic players and a finite set of actions.
 
\citet{KoeScaTom:MOR2025} extended this framework to nonatomic Bayesian games by studying the notions of correlated and Bayes correlated Wardrop equilibria. 
They showed that Bayes correlated Wardrop equilibria emerge as limits of action flows induced by Bayes correlated equilibria \citep{BerMor:TE2016} in finite-player games.\footnote{In \citet{KoeScaTom:MOR2025} players minimize cost functions, whereas here  players maximize payoff functions, which are simply the negatives of cost functions.}
Related nonatomic versions of correlated and Bayes correlated equilibria were studied, \eg by \citet{DiaMitRusSai:WINE2009,TavTen:mimeo2018}, and \citet{ZhuSav:IEEETCNS2022}.
\citet{AceMakMalOzd:OR2018} showed the relevance of information design in routing games, and proved that some classical paradoxes in routing games can be expressed in terms of  negative value of information.
Among the recent contributions on the topic, the reader is referred, for instance, to  \citet{MasLan:IEEETCNS2022,CiaAmbCom:IEEECDC2023,AmbCiaCom:IFACPOL2024}, and \citet{GriHoeKliKog:PAAAICAI2024}.
In a concurrent work focusing on congestion games and using  distributional approach,   \citet[chapter~5]{Mas:THESIS2025} considers  infinite information structures and obtains implementation results  related to ours.

To the best of our knowledge, our full implementation results with direct information structures have no analog in the literature on information design, where---even for binary action games---full implementation typically relies on complex indirect information structures, often with infinite type spaces. 
Existing tools and recent results in this area include, for example, \citet{MatPerTen:JPE2020,DwoPav:E2022,Hos:IER2022,InoPav:TE2025,LiSonZha:JET2023}, and \citet{MorOyaTak:E2024}. 
The literature on information design is surveyed in \citet{BerMor:JEL2019}, \citet{Kam:ARE2019} and \citet{For:AES2020}. \citet{Dug:SIGE2017} surveys the algorithmic literature on information design.

%
%
%
%

\subsection{Organization of the Paper}
\label{suse:organization}
\cref{se:model} presents the model, including the basic game, its Bayes correlated Wardrop equilibria, information structures with finite populations and public signals within populations, and Bayesian Wardrop ($\varepsilon$-)equilibria of the basic game extended with an information structure. 
\cref{se:examples} provides two simple examples and illustrates how to implement welfare-maximizing \ac{BCWE}. 
Our partial and full implementation results are presented in  \cref{se:partial-implementation,se:full-implementation}, respectively.
In the Appendix, 
\cref{se:optimality-deterministic} identifies a class of games for which deterministic flows maximize the designer's payoff function.
\cref{se:symbols} contains a list of symbols and acronyms used throughout the paper.

%
%
%

\subsection{Notation}
\label{suse:notation}
Given any finite set $\setA$ and $\sizepop>0$, we define
\begin{equation}
\label{eq:simplex}
\simplex_{\sizepop}(\setA) \coloneqq \braces*{\flowprof=\parens*{\flow_{\act}}_{\act} \colon \flow_{\act}  \ge 0, \sum_{\act\in \setA}\flow_{\act}=\sizepop}. 
\end{equation}

For the sake of simplicity, we use the shorthand $\simplex(\setA)$ to indicate $\simplex_{1}(\setA)$, which is just the simplex of probability measures on $\setA$. 
For a compact set $\flows$, $\simplex{(\flows)}$ denotes the set of Borel probability measures on $\flows$.
For $h \in \naturals$, the symbol $\smooth^{h}$ denotes the class of functions having an $h$-th derivative that is continuous in its domain.

%
%
%
%

\section{The Model}
\label{se:model}

We consider a basic game $\game$, which is a symmetric nonatomic Bayesian game with  a finite set of actions, and a finite set of states.
Without any loss of generality, throughout the paper, the total mass of players is assumed to be $1$.
To define the basic game, we need the concept of flow profile, which is the way the mass of players is distributed over the possible actions.
Given a finite set of  actions $\actions$, a flow profile $\flowprof \coloneqq (\flow_{\act})_{\act\in \actions}\in \flows$, is a nonnegative vector in $\flows \coloneqq \simplex(\actions)$, and $\flow_{\act}$ is the flow on action $\act$.

\begin{definition}
\label{de:BANG}    
A \emph{basic game} $\game = \parens*{\actions, \states,\prior,\utilprof}$ is given by a finite set of  actions $\actions$, a finite set of states $\states$, a full-support probability distribution $\prior\in\simplex(\states)$ over states, and a profile of payoff functions $\utilprof = (\util_{\act})_{\act\in \actions}$, where, for each $\act\in\actions$, $\util_{\act} \colon \flows \times \states \to \reals$ is continuous.
\end{definition}

For every $\act\in\actions$, $\stateth\in \states$ and $\flowprof\in \flows$, $\util_{\act}(\flowprof,\stateth)$ is the payoff of action $\act$ when the state is $\stateth$ and the flow profile in that state is $\flowprof$.
A transition  probability $\outcome \colon \states\to \simplex(\flows)$  that associates a distribution over  flows to each state is called an \emph{outcome} of the game. 

The notion of \acl{BCWE} in the next definition is the analog of the notion of \acl{BCE} \citep{BerMor:TE2016} in the context of nonatomic games \citep[see][]{KoeScaTom:MOR2025}.

\begin{definition}
\label{de:BCWE}
A \acdef{BCWE} of $\game$ is an outcome $\outcome:\states\to\simplex(\flows)$ such that
\begin{equation}
\label{eq:BCWE-def}
\sum_{\stateth\in\states} \prior(\stateth) \int \flow_{\act} \util_{\act}(\flowprof,\stateth) \diff \outcome(\flowprof \mid \stateth)
\ge
\sum_{\stateth\in\states} \prior(\stateth) \int \flow_{\act} \util_{\actalt}(\flowprof,\stateth) \diff \outcome(\flowprof \mid \stateth), \quad \forall\, \act,\actalt\in\actions.
\end{equation}
\end{definition}
For each state $\stateth\in\states$, a flow profile $\flowprof$ is drawn at random according to the distribution $\outcome(\,\cdot\mid\stateth)$. 
If we divide the two sides of \eqref{eq:BCWE-def} by the ex-ante expected mass of players choosing $\act$, then the left-hand-side is proportional to the  expected payoff of playing $\act$, conditionally of being recommended $\act$, and the right-hand side is proportional to the  expected payoff of playing $\actalt$, conditionally of being recommended $\act$. 
The corresponding outcome is a \ac{BCWE} if, for every action $\act$ that has strictly positive flow in at least one state, conditionally on being assigned $\act$, the expected payoff of playing $\act$ is greater or equal to the expected payoff of playing any other action $\actalt$. 

For example, suppose that  for every $\stateth$, $\flowprof(\stateth)$ is a Wardrop equilibrium of the complete information game $(\actions,(\util_{\act}(\argdot,\stateth))_{\act\in \actions})$, \ie 
\begin{equation}
\label{eq:Wardrop}
\forall \act,\actalt\in\actions, \ \flow_{\act}(\stateth)>0\implies\util_{\act}(\flowprof(\stateth),\stateth)\geq\util_{\actalt}(\flowprof(\stateth),\stateth).   
\end{equation}
Then $\outcome(\stateth) = \dirac_{\flowprof(\stateth)}$ is a ``fully revealing'' \ac{BCWE}. 

As another example, if $\bar\flowprof$ is a Wardrop equilibrium of the average game with complete information $(\actions,(\bar\util_{\act})_{\act\in \actions})$, with $\bar \util_{\act}(\argdot) = \sum_{\stateth}\prior(\stateth)\util_{\act}(\argdot,\stateth)$, then $\outcome(\stateth) = \dirac_{\bar\flowprof}$ is a ``non-revealing'' \ac{BCWE}. 
This corresponds to an equilibrium outcome of the basic game with no information.

We now consider  finite  information structures and Bayesian Wardrop equilibria as defined in \citet{WuAmiOzd:OR2021}. Using finite information structure has the following merits: 
On one hand this avoids  measurability issues related to private signals for  uncountable sets of players.
On the other hand, the approach is flexible enough to guarantee the possibility of arbitrarily fine information diversification.

\begin{definition}
\label{de:info-structure}
An \emph{information structure} $\infostr=\parens*{\sizepopprof, \types, \prtype}$ for   
$\game$ is given by: 
\begin{itemize}
\item 
a vector $\sizepopprof = (\sizepop^{\pop})_{\pop\in\pops} \in \simplex(\pops)$, where $\pops$ is a finite set of populations and, for every $\pop\in \pops$, $\sizepop^\pop > 0$ is the size of population $\pop$; 
\item 
for each population $\pop$, a finite  set of types $\types^{\pop}$; 

\item 
a mapping $\prtype \colon \states \to \simplex(\types)$, where $\types = \times_{\pop\in\pops} \types^{\pop}$.
\end{itemize}
\end{definition}

The game $\game$ extended by information structure $\infostr$, denoted by $(\game,\infostr)$, can be described as follows:
for each state $\stateth\in\states$, a type profile $\typeprof\in\types$ is drawn with probability $\prtype(\typeprof \mid \stateth)$ and, for every $\pop$, the entire population $\pop$ observes the type $\type^{\pop} \in \types^{\pop}$. 
Then, actions are chosen.

For all $\pop\in\pops$, $\type^{\pop}\in\types^{\pop}$ and $\act\in\actions$, we let $\flowaktype$ 
denote  the flow on action $\act$ of population $\pop$, when type $\type^{\pop}$ is observed. 
The  \emph{interim flow profile of population $\pop$ given type $\type^{\pop}$} is $\flowktypeprof =\parens*{\flowaktype}_{\act\in\actions}\in \simplex_{\sizepop^{\pop}}(\actions)$. 
The list 
\begin{equation}
\label{eq:IFP}
\widehat{\flowprof}=\parens*{\flowktypeprof}_{\pop,\type^{\pop}} \in \bigtimes_{\pop\in\pops}\bracks*{\bigtimes_{\type^{\pop}\in\types^{\pop}} \simplex_{\sizepop^{\pop}}(\actions)},
\end{equation}
of  interim flows profiles of all populations is called \emph{interim flow profile}.  
For each type profile  $\typeprof=(\type^{\pop})_{\pop\in\pops}$, we let $\flow_{\act}(\typeprof) = \sum_{\pop\in\pops} \flowaktype$ be the \emph{interim total flow on action $\act$} and $\flowprof(\typeprof) = \parens*{\flow_{\act}(\typeprof)}_{\act\in\actions}\in \simplex(\actions)$ be the \emph{interim total flow profile}, given $\typeprof$.
We let  $\prob(\type^{\pop})=\sum_{\stateth\in\states}\sum_{\typeprof^{-\pop}\in\types^{-\pop}}  \prior(\stateth) \prtype(\type^{\pop},\typeprof^{-\pop} \mid \stateth)$ denote the total probability of type $\type^{\pop}$  induced by $\prior$ and $\prtype$, where $\types^{-\pop} \coloneqq \times_{\popalt\neq \pop} \types^{\popalt}$. 
Under the interim flow profile $\widehat\flowprof$, the payoff of population $\pop$, given type $\type^\pop$ and action $\act$, such that $\flowaktype>0$, is
\begin{equation}
\label{eq:caktk}  
\sum_{\stateth\in\states}\sum_{\typeprof^{-\pop}\in\types^{-\pop}}  
\frac{\prior(\stateth) \prtype(\type^{\pop},\typeprof^{-\pop} \mid \stateth)}{\prob(\type^{\pop})}
\util_{\act}(\flowprof(\typeprof),\stateth)
\eqqcolon \util_{\act}^{\pop,\type^\pop}(\widehat\flowprof).
\end{equation}
The list $(\util_{\act}^{\pop,\type^\pop}(\widehat\flowprof))_{\pop,\type^\pop,\act}$ is called an \emph{interim payoff profile}.

The following definition provides an analog of the usual  Bayesian Nash ($\varepsilon$-)equilibrium  for nonatomic games.

\begin{definition} 
\label{de:BWepsE}
For $\varepsilon\ge0$, a \acdef{BWepsE} of $(\game,\infostr)$  is an interim flow  profile $\widehat\flowprof=\parens*{\flowktypeprof}_{\pop,\type^{\pop}}$ such that, for all $\pop\in\pops$, for all $\type^{\pop}\in\types^{\pop}$ with $\prob(\type^{\pop})>0$, and for all $\act,\actalt\in\actions$, if $\flow_{\act}^{\pop}(\type^{\pop}) > 0$, then
\begin{equation}
\label{eq:epsBWE-cond-1}
\sum_{\stateth\in\states}\sum_{\typeprof^{-\pop}\in\types^{-\pop}}  
\frac{\prior(\stateth) \prtype(\type^{\pop},\typeprof^{-\pop} \mid \stateth)}{\prob(\type^{\pop})}
\util_{\act}(\flowprof(\typeprof),\stateth) 
\geq
\sum_{\stateth\in\states}\sum_{\typeprof^{-\pop}\in\types^{-\pop}}  
\frac{\prior(\stateth) \prtype(\type^{\pop},\typeprof^{-\pop} \mid \stateth)}{\prob(\type^{\pop})}
\util_{\actalt}(\flowprof(\typeprof),\stateth)-\varepsilon.
\end{equation}
\end{definition}

That is, for each population $\pop$ that observes type $\type^{\pop}$, a positive mass  chooses action $\act$ only if, up to $\varepsilon$, it yields the highest  conditional  expected payoff. 
Notice that this is equivalent to the following inequality for all $\pop\in\pops$,  all $\type^{\pop}\in\types^{\pop}$  and  all $\act,\actalt\in\actions$ such that $\flow_{\act}^{\pop}(\type^{\pop}) > 0$:

\begin{equation}
\label{eq:epsBWE-cond}
\sum_{\stateth\in\states}\sum_{\typeprof^{-\pop}\in\types^{-\pop}}  
\prior(\stateth) \prtype(\type^{\pop},\typeprof^{-\pop} \mid \stateth)
\util_{\act}(\flowprof(\typeprof),\stateth) 
\geq
\sum_{\stateth\in\states}\sum_{\typeprof^{-\pop}\in\types^{-\pop}}  
\prior(\stateth) \prtype(\type^{\pop},\typeprof^{-\pop} \mid \stateth)
\util_{\actalt}(\flowprof(\typeprof),\stateth)-\varepsilon\prob(\type^{\pop}).
\end{equation}
A \acdef{BWE} is  a \acl{BWepsE} with $\varepsilon=0$.

An  outcome $\outcome \colon \states \to \simplex(\flows)$ is a \emph{\acl{BWepsE} outcome} of   $(\game,\infostr)$  if  there exists a \acl{BWepsE}  $\parens*{\flowktypeprof}_{\pop,\type^{\pop}}$  such that  for every $\stateth\in\states$ and for every $\flowprofz\in\flows$,
\begin{equation}
\label{eq:mu-pi}
\outcome(\flowprofz \mid \stateth) = \sum_{\typeprof \colon \flowprof(\typeprof)=\flowprofz} \prtype(\typeprof \mid \stateth).
\end{equation}
We let $\BWespE(\infostr)$ be the set of \acl{BWepsE} outcomes of   $(\game,\infostr)$.

An information structure is \emph{direct} if $\types^{\pop}=\actions$ for all $\pop\in\pops$. 
The associated \emph{obedient}  interim flow profile is the one obtained when population $\pop$ observes 
$\type^{\pop}=\act$, interprets it as an action recommendation, and obeys it by choosing $\act$. 
For every direct information structure $\infostr$, we let  $\breve\flowprof_{\infostr}$ denote the obedient interim flow profile, defined as 
$\breve\flowprof_{\infostr, \act}^{\pop}(\type^{\pop}) \coloneqq \sizepop^{\pop}\cdot\ind\braces*{\type^{\pop}=\act}$, for all $\pop\in\pops,\type^{\pop}\in\types^{\pop},\act\in\actions$. The associated outcome is denoted by $\breve\outcome_{\infostr}$.

%
%
%
%

\section{Examples}
\label{se:examples}

The following examples illustrate the notion of \ac{BCWE} in two simple games and show how to implement the total welfare-maximizing \ac{BCWE} as an obedient \ac{BWE} of the game extended with a direct information structure. 
The examples also show that restricting attention to public information or deterministic flows is not without loss of generality in terms of efficiency.

\begin{example}[Platform access with network effects and congestion]
\label{ex-El-Farol}
In the following game, studied by 
\citet{MitSaaSai:SAGT2013}, there is a single state, and the action set is $\actions = \braces*{\act,\actalt}$. 
The payoff functions are:  
\begin{equation}
\label{eq:cost-el-farol}
\util_{\act}(\flowprof) = 1, \quad 
\util_{\actalt}(\flowprof) = \min\braces*{4\flow_{\actalt},4-4\flow_{\actalt}}.
\end{equation}
\cref{eq:cost-el-farol} says that action $\actalt$ generates positive payoff externalities when $\flow_{\actalt} < 1/2$ and negative payoff externalities when $\flow_{\actalt} > 1/2$. 
This can be interpreted as a platform access decision, where action $\actalt$ corresponds to joining a platform whose value increases with moderate participation due to network effects but deteriorates when participation becomes too high because of congestion, while action $\act$ represents a safe outside option with constant payoff 1.
This game without information admits three \ac{BWE} outcomes: $(1,0)$, $(3/4,1/4)$, and $(1/4,3/4)$, all of which yield a total payoff of $1$. 
Public information induces a mixing over these outcomes, so the expected total payoff remains $1$.
Consider now the flows $\flowprofz = (1,0)$ and $\flowprofw=(1/2,1/2)$. 
The outcome $\outcome$, defined by  
\begin{equation}
\label{eq:CWE-el-farol}
\outcome(\flowprofz)=\frac{1}{3},\quad
\outcome(\flowprofw)=\frac{2}{3},
\end{equation}
is a \ac{BCWE}, since it satisfies the equilibrium constraints:  
\begin{equation}
\label{eq:equilibrium-constraints}
\begin{split}
\frac{2}{3} =
\frac{1}{3} \flowz_{\act} \util_{\act}(\flowprofz) +
\frac{2}{3} \floww_{\act} \util_{\act}(\flowprofw) 
&\ge
\frac{1}{3} \flowz_{\act} \util_{\actalt}(\flowprofz) +
\frac{2}{3} \floww_{\act} \util_{\actalt}(\flowprofw) =
\frac{2}{3},\\
\frac{2}{3} =
\frac{1}{3} \flowz_{\actalt} \util_{\actalt}(\flowprofz) +
\frac{2}{3} \floww_{\actalt} \util_{\actalt}(\flowprofw) 
&\ge
\frac{1}{3} \flowz_{\actalt} \util_{\act}(\flowprofz) +
\frac{2}{3} \floww_{\actalt} \util_{\act}(\flowprofw) =
\frac{1}{3}.
\end{split}
\end{equation}
The total payoff of this \ac{BCWE} is $4/3$, which is higher than the total payoff with public information. 
From the platform’s perspective, private recommendations are used to correlate users’ decisions so that the first-best participation level $(1/2,1/2)$ is achieved with probability $2/3$, while the inefficient outcome in which all players choose $\actalt$ never occurs.

The welfare-maximizing \ac{BCWE} is (partially) implemented as a \ac{BWE} with the following direct information structure: there are two populations of size $\sizepop^{1} = \sizepop^{2} = 1/2$, and
\begin{equation}
\label{eq:pi-ab}   
\prtype(\act,\act) = \prtype(\act,\actalt) = \prtype(\actalt,\act) = \frac{1}{3}.
\end{equation}
Operationally, the platform can achieve this outcome by partitioning users into two equally sized groups and privately recommending access to at most one group at a time.
\label{page:IS-elfarol}
\end{example}

\begin{example}[Entry on a market with uncertain production cost]
\label{ex:cournot-binary}

We consider a  game with binary action set $\actions=\braces*{\act,\actalt}$, where action $\act$ indicates staying out of the market, and action $\actalt$ indicates entering the market. 
Firms on the market produce one unit each and sell it at a market price decreasing with total production. 
The state $\stateth$ is the common production cost, which can take values $0$ or $1$. 
This game is a particular instance of the Cournot model described in \cref{se:potental-gammes}.

So we let $\states=\braces*{0,1}$, and for $\proba\in(0,1)$, 
\begin{equation}
\label{eq:prior-state}
\prior(0)=1-\proba,
\quad\text{and}\quad
\prior(1)=\proba. 
\end{equation}

Let us consider deterministic outcomes, where the flow profile in state $\stateth$ is some $\flowprof(\stateth)=\parens*{\flow_{\act}(\stateth),\flow_{\actalt}(\stateth)} = \parens*{1-\flow_{\actalt}(\stateth),\flow_{\actalt}(\stateth)}\in\flows$.
The inverse demand depends linearly on total production and is given by 
\begin{equation}
\label{eq:inverse-demand} 
\invdem(\flow_{\actalt}(\stateth)) = 1-\flow_{\actalt}(\stateth).
\end{equation}
Firms maximize profits, which are
\begin{equation}
\label{eq:profit}  
\begin{split}
\util_{\act}(\flowprof,\stateth)
&=
0,\\
\util_{\actalt}(\flowprof,\stateth)
&=
\invdem(\flow_{\actalt}(\stateth)) - \stateth. 
\end{split}
\end{equation} 
The total profit in state $\stateth$ is 
\begin{equation}
\label{eq:total-profit}
\flow_{\act}(\stateth)\util_{\act}(\flowprof(\stateth),\stateth) + \flow_{\actalt}(\stateth)\util_{\actalt}(\flowprof(\stateth),\stateth)
=\flow_{\actalt}(\stateth)(1 - \flow_{\actalt}(\stateth) - \stateth).
\end{equation}
The expected total profit is
\begin{equation}
\label{eq:expected-total-profit}
(1-\proba)\,\flow_{\actalt}(0)(1 - \flow_{\actalt}(0)) - \proba\,\flow_{\actalt}(1)\flow_{\actalt}(1).
\end{equation}
The maximizer over $[0,1]^{2}$ is $\parens*{\flow_{\actalt}(0)=1/2,\flow_{\actalt}(1)=0}$, with expected total profit equal to $(1-\proba)/4$.

For deterministic outcomes, the \ac{BCWE} equilibrium constraints reduce to:
\begin{align}
\label{eq:OB0-cournot}
(1-\proba)(1-\flow_{\actalt}(0))^{2}
&\le \proba\,\flow_{\actalt}(1)(1-\flow_{\actalt}(1)),\\
\label{eq:OB1-cournot}
(1-\proba)\flow_{\actalt}(0)(1-\flow_{\actalt}(0))
&\ge \proba\,(\flow_{\actalt}(1))^{2}.
\end{align}
Maximizing expected total profit subject to
\cref{eq:OB1-cournot,eq:OB0-cournot} yields a unique solution with
\eqref{eq:OB0-cournot} binding and \eqref{eq:OB1-cournot} slack:
\begin{equation}
\label{eq:unique-sol}
\flow_{\actalt}^{*}(0)=1-\frac{\sqrt{\proba}}{2},
\quad
\flow_{\actalt}^{*}(1)=\frac{1-\sqrt{\proba}}{2}.    
\end{equation}

The associated expected total profit is $\parens*{\sqrt{\proba}-\proba}/2$, which is strictly smaller than the unconstrained benchmark but strictly larger than the expected total profit achievable under public information, \ie $0$. 
By \cref{pr:BDWE-cost}, restricting attention to deterministic flows is without loss of optimality in this binary-action environment, so $(\flow_{\actalt}^{*}(0),\flow_{\actalt}^{*}(1))$ is the total profit-maximizing \ac{BCWE} outcome.

Similarly to the previous example, as long as $\sqrt{\proba}$ is a rational number, this deterministic \ac{BCWE} can be implemented as an obedient \ac{BWE} using a direct information structure with finitely many populations, by recommending action $\actalt$ to a fraction $\flow_{\actalt}^{*}(\stateth)$ of firms in each state $\stateth$. 
For example, if $\proba=1/4$, then 
\begin{equation}
\label{eq:}
\flow_{\actalt}^{*}(0)=1-\frac{\sqrt \proba}{2}=\frac{3}{4},
\qquad
\flow_{\actalt}^{*}(1)=\frac{1-\sqrt \proba}{2}=\frac{1}{4}.    
\end{equation}

\label{page:IS-Cournot}
Consider four populations of equal size $\sizepop^{\pop}=1/4$ for every $\pop$, and a direct information structure such that, in state $0$, exactly three populations are recommended action $\actalt$ and one population is recommended action $\act$, whereas in state $1$ exactly one population is recommended action $\actalt$ and three populations are recommended action $\act$ (uniformly over the choice of the recommended populations). 
This induces the deterministic outcome $(\flow_{\actalt}(0),\flow_{\actalt}(1))=(3/4,1/4)$. 
As shown in \cref{se:full-implementation}, this total profit-maximizing outcome is fully implemented in this game because it is a concave potential game. 
Moreover, when $\sqrt{\proba}$ is not rational, the total profit-maximizing outcome can be implemented with arbitrary precision by using a sufficiently large number of populations.
\end{example}

%
%
%
%

\section{Partial Implementation}
\label{se:partial-implementation}

It is immediate  that any \ac{BWE} outcome, for any possible  information structure (direct or indirect), is  in the set of \acp{BCWE}. 
This simple observation explicitly illustrates how to construct \emph{some} \ac{BCWE}: choose an information structure $\parens*{\sizepopprof, \types, \prtype}$ for disclosing information, and  choose a \ac{BWE} of the basic game extended with this information structure.

The next theorem establishes the other direction of the revelation principle (\citealp{Mye:JME1982}, \citealp{BerMor:TE2016}), at least approximately. 
Specifically, it shows that if $\outcome$ is a \ac{BCWE} with finite support and rational flows, then there exists a direct information structure such that $\outcome$ is implemented as an obedient \ac{BWE} outcome. 
If the support of $\outcome$ is not finite or if some flows in the support are not rational, then $\outcome$ can be approximately implemented as a Bayesian Wardrop $\varepsilon$-equilibrium outcome for every $\varepsilon > 0$.

\begin{theorem}
\label{thm:partialBI}
Let $\outcome$ be a \ac{BCWE} 
and $\descost \colon \flows \times \states \to \reals$ be a continuous function. 
For all 
$\varepsilon>0$, there exists a direct information structure $\infostr$ such that the obedient interim flow profile  $\breve\flowprof_{\infostr}$  is a  \acl{BWepsE} of   $(\game,\infostr)$ and
\begin{equation}
\label{eq:eps-approx-BW}
\abs*{\int  \descost(\flowprof,\stateth) \diff \outcome(\flowprof\mid \stateth)
-
\int  \descost(\flowprof,\stateth) \diff \breve\outcome_{\infostr}(\flowprof\mid \stateth)}
\leq\varepsilon, \quad \forall \stateth\in\states.
\end{equation}
Further,  if $\outcome$ has finite support and if all flows in the support of $\outcome$ are rational, then there exists an information structure $\infostr$ such that  $\varepsilon = 0$ and $\outcome=\breve\outcome_{\infostr}$.
\end{theorem}

In \cref{thm:partialBI}, the function $\descost$ represents any possible designer welfare function, and \cref{eq:eps-approx-BW} means that for every state, the interim expected payoff of the designer under outcome $\breve\outcome_{\infostr}$ is $\varepsilon$-close to the interim expected payoff under the \ac{BCWE} $\outcome$. 
The proof of the theorem also shows that, for every $\lipschitz> 0$,  the same $\infostr$ can be chosen for all $\lipschitz$-Lipschitz functions $\descost$. 

\begin{proof}[Proof of \cref{thm:partialBI}]  
\citet[Lemma~2]{KoeScaTom:MOR2025} prove  that $\ac{BCWE}$ with finite support are dense within the set of \aclp{BCWE}. 
So, it is enough to prove the result for $\outcome$ with finite support.
Fix $\varepsilon>0$ and a \ac{BCWE}   $\outcome$  with finite support $\suppflows \subset \flows$. Let $\actions^*$ be the set of actions $a$ actually used under $\outcome$, \ie such that
\begin{equation}  
\sum_{\stateth\in\states}\sum_{\flowprof\in\suppflows}\prior(\stateth)\outcome(\flowprof\mid\stateth)\flow_{\act}>0.
\end{equation}
Let then
\begin{equation}
\varepsilon_{0}=\min\Big\{\sum_{\stateth\in\states}\sum_{\flowprof\in\suppflows}\prior(\stateth)\outcome(\flowprof\mid\stateth)\flow_{\act} \text{ s.t. } \act\in\actions^*\Big\}>0.
\end{equation}

Fix an arbitrary $\error>0$ and approximate flows in $\suppflows$ by   vectors with  rational components as follows. 
There exists $\nat\in\naturals_+$, such that for every $\flowprof\in\suppflows$ and every $\act\in\actions^*$,  there exists $\napproxim(\flow_{\act})\in\naturals$,  such that,
\begin{equation}
\label{eq:approx-rational}
\sum_{\act\in\actions^*} \napproxim(\flow_{\act}) = \nat
\quad\text{and}\quad
\forall\act\in\actions^*,
\abs*{\frac{\napproxim(\flow_{\act})}{\nat} - \flow_{\act}} \le \error.
\end{equation}
For each $\flowprof\in\suppflows$, let $\widetilde\flowprof$ denote the flow profile $(\widetilde\flow_{\act})_{\act\in\actions^*}$ with $\widetilde\flow_{\act}= \napproxim(\flow_{\act})/\nat$.

We construct a direct information structure $\infostr$ as follows. 
There are $|\pops|=\nat$ populations and each population $\pop$  has the same size $\sizepop^{\pop} = 1/\nat$.
For each population $\pop$, the set of types is $\types^{\pop}=\actions^*$. 
Conditionally on state $\stateth$, a flow profile $\flowprof\in\suppflows$ is drawn with probability $\outcome(\flowprof\mid\stateth)$, then  action $\act$ is recommended to a subset of   $\napproxim(\flow_{\act})$ populations chosen uniformly, so that, conditionally on $\flowprof$ and $\stateth$, the probability that population $\pop$ is recommended $\act$ is $\widetilde\flow_{\act}$.
We let $\breve\outcome_{\infostr}$ be the induced obedient outcome which puts probability $\outcome(\flowprof\mid\stateth)$ on $\widetilde\flowprof$ in state $\stateth$.

For each population $\pop$, the  conditional expected payoff  of playing $\actalt$ when  $\act$ is recommended is
\begin{equation}
\label{eq:nonnormalized-expected-cost}
\sum_{\stateth\in\states}\sum_{\flowprof\in\suppflows}
\frac{\prior(\stateth)\outcome(\flowprof\mid\stateth)\widetilde\flow_{\act}}{\Qmeas(\act)}
\util_{\actalt}(\widetilde\flowprof,\stateth),
\end{equation}
with 
\begin{equation}
\label{eq:Q-k}  
\Qmeas(\act)=\sum_{\stateth\in\states}\sum_{\flowprof\in\suppflows}\prior(\stateth)\outcome(\flowprof\mid\stateth)\widetilde\flow_{\act}.
\end{equation}
For each pair of actions $(\act,\actalt)$, the mapping $\flowprof \mapsto \flow_{\act}\util_{\actalt}(\flowprof,\stateth)$ is uniformly continuous on the compact set $\flows$. 
Thus, there exists a common modulus of continuity $\modcont(\argdot)$ with $\lim_{\error\searrow 0}\modcont(\error)=0$ such that, for all $\stateth\in\states$, for all $(\act,\actalt)\in \actions\times \actions$, for all $(\flowprof,\flowprofz)\in\flows\times\flows$,
\begin{equation}
\label{eq:ModulusCont} 
\abs*{\flow_{\act}\util_{\actalt}(\flowprof,\stateth)- \flowz_{\act}\util_{\actalt}(\flowprofz,\stateth)}\leq\modcont(\max_{\actalt}\abs{\flow_{\actalt}-\flowz_{\actalt}}).
\end{equation}
Since $\outcome$ is \ac{BCWE}, for all $\act,\actalt\in\actions$,
\begin{equation}
\label{eq:BCWE-cond}
\sum_{\stateth\in\states}\sum_{\flowprof\in\suppflows}
\prior(\stateth)\outcome(\flowprof\mid\stateth) \flow_{\act} \util_{\act}(\flowprof,\stateth)
\ge
\sum_{\stateth\in\states}\sum_{\flowprof\in\suppflows}
\prior(\stateth)\outcome(\flowprof\mid\stateth) \flow_{\act} \util_{\actalt}(\flowprof,\stateth).
\end{equation}
From \cref{eq:BCWE-cond,eq:approx-rational,eq:ModulusCont} it follows that
\begin{equation}
\label{eq:ineq-mod-cont}
\sum_{\stateth\in\states}\sum_{\flowprof\in\suppflows}\prior(\stateth)\outcome(\flowprof\mid\stateth)\widetilde\flow_{\act}\util_{\act}(\widetilde\flowprof,\stateth)
\ge
\sum_{\stateth\in\states}\sum_{\flowprof\in\suppflows}\prior(\stateth)\outcome(\flowprof\mid\stateth)\widetilde\flow_{\act}\util_{\actalt}(\widetilde\flowprof,\stateth)-2\modcont(\error).
\end{equation}
Thus
\begin{equation}
\label{eq:ineq-mod-cont-2}
\sum_{\stateth\in\states}\sum_{\flowprof\in\suppflows}
\frac{\prior(\stateth)\outcome(\flowprof\mid\stateth)\widetilde\flow_{\act}}{\Qmeas(\act)} 
\util_{\act}(\widetilde\flowprof,\stateth)
\ge
\sum_{\stateth\in\states}\sum_{\flowprof\in\suppflows}
\frac{\prior(\stateth)\outcome(\flowprof\mid\stateth)\widetilde\flow_{\act}}{\Qmeas(\act)} 
\util_{\actalt}(\widetilde\flowprof,\stateth)
-
\frac{2\modcont(\error)}{\Qmeas(\act)}.
\end{equation}
From \cref{eq:approx-rational} we have
\begin{equation}
\label{eq:P-k-a}    
\Qmeas(\act)=
\sum_{\stateth\in\states}\sum_{\flowprof\in\suppflows}\prior(\stateth)\outcome(\flowprof\mid\stateth)\widetilde\flow_{\act}
\geq
\sum_{\stateth\in\states}\sum_{\flowprof\in\suppflows}\prior(\stateth)\outcome(\flowprof\mid\stateth)(\flow_{\act}-\error)
\geq
\varepsilon_{0}-\error\geq \frac{\varepsilon_{0}}{2},
\end{equation}
for $\error\leq  \varepsilon_{0}/2$. 
It follows that \begin{equation}
\label{eq:ineq-omega-eta-eps}
\frac{2\modcont(\error)}{\Qmeas(\act)}\leq\frac{2\modcont(\error)}{\varepsilon_{0}/2}.
\end{equation} 
Therefore, by choosing $\error$ such that the right-hand side of \cref{eq:ineq-omega-eta-eps} is smaller than $\varepsilon$, we obtain a \ac{BWepsE}. 
 
Now, take a  continuous $\descost \colon \flows\times \states\to \reals$. 
Since $\flows\times \states$ is compact, $\descost$ is uniformly continuous and admits a modulus of continuity $\modcont_{\descost}$ such that for all $\stateth\in\states$, 
for all $(\flowprof,\flowprofz) \in \flows\times\flows$,
\begin{equation}
\label{eq:ModulusCont-2} 
\abs*{\descost(\flowprof,\stateth)- \descost(\flowprofz,\stateth)}\leq\modcont_{\descost}(\max_{\actalt}\abs{\flow_{\actalt}-\flowz_{\actalt}}).
\end{equation}
For each $\stateth\in\states$, we have
\begin{equation}
\label{eq:ModContF}
\begin{split}
\abs*{\sum_{\flowprof\in\suppflows}\outcome(\flowprof\mid\stateth)\descost(\widetilde\flowprof,\stateth) - \sum_{\flowprof\in\suppflows}\outcome(\flowprof\mid\stateth) \descost(\flowprof,\stateth)}
&\leq
\sum_{\flowprof\in\suppflows}\outcome(\flowprof\mid\stateth)
\abs*{\descost(\widetilde\flowprof,\stateth) - \descost(\flowprof,\stateth)}\\
&\leq
\sum_{\flowprof\in\suppflows}\outcome(\flowprof\mid\stateth)\modcont_{\descost}(\error)\\
&\leq
\modcont_{\descost}(\error),
\end{split}    
\end{equation}
which is less or equal to $\varepsilon$ for $\error$ small. 

If  flows in $\suppflows$ have rational coefficients with common denominator $\nat$, then we can choose $\widetilde\flowprof=\flowprof$ and $\error=0$, and we have a Bayesian Wardrop $0$-equilibrium. 
Finally, notice that if $\descost$ is $\lipschitz$-Lipschitz (for the sup norm on flows), then we can choose $\modcont_{\descost}(\error)=\lipschitz\error$. 
Thus, for given $\varepsilon$, the choice of $\error$ and $\infostr$ is the same across all such $\descost$.
\end{proof}

%
%
%
%

\section{Potential Games and Full Implementation}
\label{se:full-implementation}
In this section, we study full implementation of \ac{BCWE} in games with a concave potential.

%
%
%
%

\subsection{Potential Games}
\label{se:potental-gammes}

\begin{definition}
\label{de:potential}
A  basic game  
$\game = \parens*{\actions, \states,\prior,\utilprof}$ is called a \emph{potential game} if for every $\stateth\in\states$, there exists an open superset  
$\widetilde{\flows}$
of $\flows$ in $\reals^\actions$ and a continuously differentiable  function  $\potential_{\stateth} \colon \widetilde{\flows} \to \reals$  such that, for every $\act\in\actions$ and $\flowprof\in\flows$,
\begin{equation}
\label{eq:potential-incomp}
\frac{\partial\potential_{\stateth}(\flowprof)}{\partial\flow_{\act}} = \util_{\act}(\flowprof,\stateth).
\end{equation}
 The function $\potential_{\stateth}$ is called \emph{the potential  in state $\stateth$}.
\end{definition}

As a first example, consider symmetric congestion games, which are defined by:
a finite set of \emph{resources} $\edges$; a set of actions  $\actions \subseteq 2^{\edges}$; for each resource $\edge\in\edges$, a continuous cost function $\cost_{\edge} \colon \reals_{+}\times\states \to \reals_{+}$. 
The payoff of choosing action $\act \in \actions$ in state $\stateth$ is the negative of the sum of the costs of all resources used in $\act$, 
\begin{equation}
\label{eq:congestion}
\util_{\act}(\flowprof,\stateth)\coloneqq-\sum_{\edge\in\act} \cost_{\edge}(\load_{\edge},\stateth), 
\end{equation}
where $\load_{\edge} \coloneqq  \sum_{\act\in \actions, \act\ni\edge} \flow_{\act}$ is the \emph{load} on resource $\edge$ induced by the flow $\flowprof$. 
The potential function of a congestion game is
\begin{equation}
\label{eq:potential-congestion}
\potential_{\stateth}(\flowprof) =-\sum_{\edge\in\edges} \int_{0}^{\load_{\edge}} \cost_{\edge}(t,\stateth) \diff t.
\end{equation}
When the cost function of each resource is increasing for every state,  the potential is concave. 
In \emph{singleton congestion games} each action uses a single resource, \ie $\actions=\{\{\edge\}, \edge\in\edges\}$.
In this class of games, if the cost functions are strictly increasing, the potential is strictly concave.

The following lemma shows that, unlike what happens in games with a finite set of players  \citep[see][]{MonSha:GEB1996}, not all potential games are congestion games. 

\begin{lemma}
\label{lem:potentialcongestion} 
Fix an arbitrary state $\stateth$. 
A game with a potential $\potential_{\stateth} \in \smooth^{3}$ is a congestion game only if
\begin{equation}
\label{eq:partial-three}
\frac{\partial^{3}\potential_{\stateth}(\flowprof)}{\partial\flow_{\act}^{2}\partial\flow_{\actalt}} 
=\frac{\partial^{3}\potential_{\stateth}(\flowprof)}{\partial\flow_{\actalt}^{2}\partial\flow_{\act}}.    
\end{equation}
\end{lemma}

\begin{proof}
Consider a potential game with a potential function $\potential_{\stateth} \in \smooth^{3}$, and suppose that it is a congestion game, with cost functions $\cost_{\edge}(\argdot,\stateth)$, which are then twice  continuously differentiable.  
Combine  \cref{eq:potential-incomp,eq:congestion} to get,
\begin{equation}
\label{eq:potential-congestion2}
\frac{\partial\potential_{\stateth}(\flowprof)}{\partial\flow_{\act}} =-\sum_{\edge\in\act} \cost_{\edge}(\load_{\edge},\stateth).
\end{equation}
This implies that 
\begin{equation}
\label{eq:potential-congestion3}
\frac{\partial^{2}\potential_{\stateth}(\flowprof)}{\partial\flow_{\act}\partial\flow_{\actalt}} =-\sum_{\edge\in\act\cap\actalt} \cost'_{\edge}(\load_{\edge},\stateth)
=\frac{\partial^{2}\potential_{\stateth}(\flowprof)}{\partial\flow_{\actalt}\partial\flow_{\act}} .
\end{equation}
Observe  that $\load_{\edge} =  \sum_{\act'\in \actions, \act\ni\edge} \flow_{\act'}$.
Because we consider only resources $\edge\in\act\cap\actalt$, the flow $\flow_{\act}$ or $\flow_{\actalt} $ only appears in the sum $\flow_{\act}+\flow_{\actalt} $. 
Therefore, taking derivative once more with respect to either $\flow_{\act}$ or $\flow_{\actalt} $ yields then the same result. 
\end{proof}

To go beyond congestion games, we consider the class of \emph{aggregative games}.
This class was implicitly introduced by \citet{DubMasShu:JET1980}, formalized by \citet{Cor:MSS1994}, and recently studied by \citet{Lah:DGA2017} for nonatomic players and a finite set of actions. 
Aggregative games are also given as examples of potential games in  \citet[chapter 3]{San:MIT2010}. 

An aggregative game is defined as follows. 
Let the action set $\actions$ be a finite subset of $\reals_{+}$, so actions can be seen as quantities. 
For $\flow\in\Delta(\actions)$, let the aggregate quantity be $\quantity(\flowprof)=\sum_{\act\in \actions}\act\flow_{\act}$. 
Each payoff $\util_{\act}$ depends only on the chosen action $\act$  and on the aggregate quantity, \ie for each $\act\in\actions$ and for each $\stateth\in\states$, there exists a function $\gfunc_{\act}(\argdot,\stateth)$ such that
$\util_{\act}(\flow,\stateth)=\gfunc_{\act}(\quantity(\flowprof), \stateth)$.

Consider payoffs of the following form:
\begin{equation}
\label{eq:potential-congestion5}
\forall \act\in\actions, \util_{\act}(\flow,\stateth)=\act \invdem(\quantity(\flowprof), \stateth)-\const_{\act}(\stateth),
\end{equation}
where $\invdem(\argdot, \stateth)$ is a function of the aggregate quantity $\quantity(\flowprof)$, and $\const_{\act}(\stateth)$ is a number that does not depend on $\flowprof$. 
This game  admits the potential:
\begin{equation}
\label{eq:potential-aggreg}   \potential_\stateth(\flowprof)=\int_{0}^{\quantity(\flowprof)}\invdem(t,\stateth) \diff t-\sum_{\act} \const_{\act}(\stateth)\flow_{\act}. 
\end{equation}
If $\invdem(\argdot,\stateth)$ is a decreasing function of the quantity, then the potential is concave.\footnote{\citet{San:MIT2010} and \citet{Lah:DGA2017} also consider the case of increasing $\invdem(\argdot,\stateth)$, which models public good contribution games where the value of the good is increasing with the total contribution, and players incur contribution costs. 
In this case, the  potential is convex.} 
This potential can be interpreted as a  demand function and the game is a \emph{Cournot game},  where $\const_{\act}(\stateth)$ is the production cost of quantity $\act$. 
In this model, both  the demand and  the production costs may be uncertain. 
In \cref{ex:cournot-binary} only production costs are uncertain and the potential is strictly concave.

With the help of \cref{lem:potentialcongestion}, assuming that $\invdem(\argdot,\stateth)$ is twice continuously differentiable, we can check whether the game given by \cref{eq:potential-congestion5} is a congestion game or not.
Suppose that $\invdem''(\argdot, \stateth)>0$ and assume that there exist two distinct actions $\act,\actalt \in \actions\setminus\{0\}$. 
Then
\begin{equation}
\label{eq:potential-check}  
\frac{\partial^{3}\potential_{\stateth}(\flowprof)}{\partial\flow_{\act}^{2}\partial\flow_{\actalt}} 
=\act^{2}\actalt \invdem''(\quantity(\flowprof), \stateth)\neq \act \actalt^{2} \invdem''(\quantity(\flowprof),\stateth)=
\frac{\partial^{3}\potential_{\stateth}(\flowprof)}{\partial\flow_{\actalt}^{2}\partial\flow_{\act}}.
\end{equation}
Therefore, the game is not a congestion game. Take for instance a demand function given by  $\invdem(q,\stateth)=\kappa_{\stateth}/(\lambda_{\stateth}+q)$.

\begin{proposition}
\label{pr:potential-aggregative-games}   
A potential game with a potential $\potential_{\stateth} \in \smooth^{2}$ is aggregative if and only if \cref{eq:potential-congestion5} holds.
\end{proposition}

\begin{proof}
Take an aggregative game with payoffs $\util_{\act}(\flow,\stateth)=\gfunc_{\act}(\quantity(\flowprof), \stateth)$ and suppose it admits a potential $\potential_{\stateth}$.
Then, 
\begin{align}
\label{eq:potential-g} 
\frac{\partial\potential_{\stateth}(\flowprof)}{\partial\flow_{\act}}
&=
\gfunc_{\act}(\quantity(\flowprof), \stateth)
\intertext{and, for all $\act,\actalt\in\actions$,}
\label{eq:potential-g-2} 
\frac{\partial^{2}\potential_{\stateth}(\flowprof)}{\partial\flow_{\actalt}\partial\flow_{\act}}
&=
\actalt \gfunc'_{\act}(\quantity(\flowprof), \stateth)
=
\act \gfunc'_{\actalt}(\quantity(\flowprof), \stateth)
=
\frac{\partial^{2}\potential_{\stateth}(\flowprof)}{\partial\flow_{\act}\partial\flow_{\actalt}}.
\end{align}
So for any two actions $\act,\actalt\neq 0$,  the two ratios 
\begin{equation}
\label{eq:ratios}  
\frac{\gfunc'_{\act}(\quantity(\flowprof), \stateth)}{\act}=\frac{\gfunc'_{\actalt}(\quantity(\flowprof), \stateth)}{\actalt}
\end{equation}
are equal to some common function  of $\quantity(\flowprof)$.
Call this function  $\invdem'(\quantity(\flowprof),\stateth)$. 
\cref{eq:potential-congestion5} then follows by integration, where $\const_{\act}(\stateth)$ is an integration constant. 
Note that for action $0$,  \cref{eq:potential-g-2} implies that  $\gfunc'_{0}(\quantity(\flowprof), \stateth)=0$, so   \cref{eq:potential-congestion5} is valid for $\act=0$ as well.
\end{proof}

Finally, notice that every game with two actions is a potential game, though the potential need not be concave (see, \eg \cref{ex-El-Farol}).

%
%
%
%

\subsection{Full Implementation}
\label{suse:full-implementation}
The next proposition shows that for every potential game with a (strictly) concave potential,  the \ac{BWE} payoff profile (outcome) of the game extended by an information structure is unique.

\begin{proposition}
\label{pr:unique-BWE}
Consider a basic game $\game$. 
If the potential is concave for each state, then for every information structure $\infostr$, all \ac{BWE} of $(\game,\infostr)$ have the same payoff profile.  
If the potential is strictly concave for each state, then for every information structure $\infostr$ there exists a unique \ac{BWE} outcome of $(\game,\infostr)$.
\end{proposition}

To prove \cref{pr:unique-BWE}, we  first show that a \ac{BWE} of $(\game,\infostr)$  is equivalent to a \acl{WE} of an auxiliary game with complete information  and multiple populations, which is also a potential game. 
If the basic game has a concave potential in each state, then the potential of the auxiliary game   is  also concave. 
Uniqueness of \acl{WE} payoffs in the auxiliary game follows from  \citet[Proposition 4]{KoeScaTom:MOR2025}.
As a consequence, we get uniqueness of \ac{BWE} payoff profiles for $(\game, \infostr)$.
If the basic game has a strictly concave potential in each state, then the potential of the auxiliary game is a strictly concave function of the interim total flow profiles $(\flowprof(\typeprof))_{\typeprof}$. 
Uniqueness of the equilibrium interim total flow profile follows from the fact that the \acl{WE} of the auxiliary multi-population game with complete information  are the  maximizers of the potential  \citep[see][proposition~3.1]{San:JET2001}.

\begin{proof}[Proof of \cref{pr:unique-BWE}]
Consider a basic $\game$  and an information structure $\infostr=\parens*{\sizepopprof, \types, \prtype}$. 
Define a multi-population auxiliary 
game with complete information where populations are indexed by $(\pop,\type^{\pop})$, for all $\pop\in\pops, \type^{\pop}\in\types^{\pop}$. 
A flow profile for population $(\pop,\type^{\pop})$ is  $\flowprof^{\pop}(\type^{\pop}) = \parens*{\flow_{\act}^{\pop}(\type^{\pop})}_{\act\in\actions}\in \simplex_{\sizepop^{\pop}}(\actions) $.
Given the interim flow profile 
$\widehat \flowprof=\parens*{\flow^{\pop}_{\act}(\type^{\pop})}_{\pop,\type^{\pop},\act}$, the payoff function of population  $(\pop,\type^{\pop})$ is
\begin{equation}
\label{eq:cost-ANF}
\Util_{\act}^{\pop,\type^{\pop}}(\widehat\flowprof)
\coloneqq \sum_{\stateth,\typeprof^{-\pop}} \prior(\stateth) \prtype\parens*{\type^{\pop},\typeprof^{-\pop}\mid\stateth} \util_{\act}
\parens*{\parens*{\sum_{\popalt\in\pops}\flow_{\actalt}^{\popalt}(\type^{\popalt})}_{\actalt\in \actions},\stateth}=\prob(\type^\pop)\util_{\act}^{\pop,\type^{\pop}}(\widehat\flowprof).
\end{equation}
A \acl{WE} of this auxiliary game is identical to a \acl{BWE} of $(\game,\infostr)$.

Suppose that in each state $\stateth$ there is a potential $\potential_{\stateth}$. 
Then, let
\begin{equation}
\label{eq:potential-theta}
\potential_{\prtype}( \widehat\flowprof)
\coloneqq \sum_{\stateth,\type}\prior(\stateth)\prtype(\typeprof\mid\stateth) \potential_{\stateth}\parens*{\parens*{\sum_{\pop}\flow_{\act}^{\pop}(\type^{\pop})}_{\act\in\actions}} = \sum_{\stateth,\type}\prior(\stateth)\prtype(\typeprof\mid\stateth)\potential_{\stateth}(\flowprof(\typeprof)),
\end{equation}
where we use the notation $\flow_{\act}(\typeprof) = \sum_{\pop\in\pops} \flow_{\act}^{\pop}(\type^{\pop})$ and 
$\flowprof(\typeprof) = \parens*{\flow_{\act}(\typeprof)}_{\act\in\actions}.$ 
The function $\potential_{\prtype}$ 
is a potential for the auxiliary game and
\begin{equation}
\label{eq:potential-auxiliary}
\frac{\partial \potential_{\prtype}(\widehat\flowprof)}{\partial \flow_{\act}^{\pop}(\type^{\pop})} = \Util_{\act}^{\pop,\type^{\pop}}(\widehat\flowprof).   
\end{equation} 
If $\potential_{\stateth}$ is strictly concave for each $\stateth$, then $\potential_{\prtype}$  is a strictly concave function of the vector
$(\flowprof(\typeprof))_{\typeprof}$. 
By  \citet[Proposition~3.1]{San:JET2001}, the equilibrium flows are the maximizers of the potential function.
Thus, there is a unique equilibrium vector $(\flowprof(\typeprof))_{\typeprof}$, and therefore a unique equilibrium outcome, given by  \eqref{eq:mu-pi}. If $\potential_{\prtype}$  is (weakly) concave, uniqueness of the equilibrium payoff profile $(\Util_{\act}^{\pop,\type^{\pop}}(\widehat\flowprof))_{\pop,\type^\pop,\act}$ follows from  \citet[proposition 4]{KoeScaTom:MOR2025}.
Uniqueness of the \ac{BWE} payoff profile follows from \cref{eq:cost-ANF}.
\end{proof} 

The next theorem is a full implementation result for games with a strictly concave potential. 
It shows that if $\outcome$ is a \ac{BCWE} with finite support and rational flows, then there exists a direct information structure such that $\outcome$ is implemented as a unique \acl{BWE} outcome. 
Furthermore, the information structure is direct and the \ac{BWE} is obedient. 
If the support of $\outcome$ is not finite or if some flows in the support are not rational, then $\outcome$ is close to the unique \ac{BWE} outcome  $\eq{\outcome}$ for some direct information structure $\infostr$.

\begin{theorem}
\label{thm:fullBI}
Let $\game$ be a basic game with a strictly concave potential in each state. 
Let  $\outcome$ be a \ac{BCWE} and let $\descost \colon \flows \times \states \to \reals$ be a continuous function. 
For all
$\varepsilon>0$,  there exists a direct information structure $\infostr$ such that the \acl{BWE} outcome $\eq{\outcome}$ of $(\game,\infostr)$ is unique and
\begin{equation}
\label{eq:diff-mu-mu-star}   
\abs*{\int  \descost(\flowprof,\stateth) \diff \outcome(\flowprof\mid \stateth)
-
\int  \descost(\flowprof,\stateth) \diff \eq{\outcome}(\flowprof\mid \stateth)} \leq \varepsilon, \quad \forall \stateth\in\states.
\end{equation}
In addition,  if  $\outcome$ has finite support with  rational flows, then we can choose $\varepsilon = 0$ and the unique  \acl{BWE} outcome $\eq{\outcome}$ is the obedient outcome  $\breve\outcome_{\infostr}$.
\end{theorem}

As before, $\descost$ represents any possible welfare function for the designer. 
Also,  for every $\lipschitz>0$, we can choose the same $\infostr$ for all $\lipschitz$-Lipschitz $\descost$.

The structure of the proof is as follows. 
Consider the construction given in the proof of  \cref{thm:partialBI}: 
Approximate $\outcome$ by a distribution over a finite set of rational flows and construct a direct information structure such that the obedient outcome  is a \acl{BWalpE} for some $\alpha>0$, and  is close to $\outcome$. 
The difficulty is that the obedient outcome  need not be an exact \acl{BWE} and that we don't know how to construct this \ac{BWE} explicitly. 
We show that \emph{any}  \acl{BWalpE} is an $\alpha$-maximizer of the potential of the game with the constructed information structure, and that all  $\alpha$-maximizers are   $\varepsilon$-close to each other for some $\varepsilon>0$.  
In particular, the unique \ac{BWE} outcome $\eq{\outcome}$ is close to the obedient outcome and thus, also close to $\outcome$.
This reasoning is summarized in the following proposition.

\begin{proposition}
\label{pr:fullBI}
Let $\game$ be a basic game with a strictly concave potential in each state.
Let  $\outcome$ be a \ac{BCWE}  and let  $\descost \colon \flows \times \states \to \reals$ be a continuous function. 
For all
$\varepsilon>0$,  there exist $\alpha>0$ and a direct information structure $\infostr$ such that the \acl{BWE} outcome $\eq{\outcome}$ of $(\game,\infostr)$ is unique, and  for all 
$\widetilde\outcome\in \BWalpE(\infostr)$,
\begin{equation}
\label{eq:diff-eps}
\abs*{\int  \descost(\flowprof,\stateth) \diff \outcome(\flowprof\mid \stateth)
-
\int  \descost(\flowprof,\stateth) \diff \widetilde\outcome(\flowprof\mid \stateth)}
 \leq \varepsilon, \quad\forall \stateth\in\states.
\end{equation}
\end{proposition}

Note that  $\eq{\outcome}$ belongs to  $\BWalpE(\infostr)$ for all $\alpha\geq 0$.

\begin{proof}[Proof of \cref{pr:fullBI}]
We take up the construction from the proof of  \cref{thm:partialBI}:  
fix $\alpha>0$, assume that $\outcome$ has finite support,   and consider the information structure $\infostr$ constructed there  with parameter $\error>0$, such that the obedient flow profile $\breve\flowprof_{\infostr}$ with outcome $\breve\outcome_{\infostr}$ is a \acl{BWalpE} and 
\begin{equation}
\label{eq:diff-alpha}    
\abs*{\int  \descost(\flowprof,\stateth) \diff \outcome(\flowprof\mid \stateth)
-
\int  \descost(\flowprof,\stateth) \diff \breve\outcome_{\infostr}(\flowprof\mid \stateth)} 
\leq \alpha, \quad\forall \stateth\in\states.
\end{equation}

\begin{claim}
\label{cl:alpMin}
Suppose that the potential is (weakly) concave in each state. For all $\alpha>0$, any \acl{BWalpE} of the game with information structure $\infostr$, is an $\alpha$-maximizer of the associated potential. 
\end{claim}

\begin{proof}
[Proof of \cref{cl:alpMin}]
Consider  an arbitrary \acl{BWalpE} $\widehat\flowprof$. 
For all $\pop\in\pops$, for all $\type^{\pop}\in\types^{\pop}$, and for all $\act,\actalt\in\actions$, if $\flow_{\act}^{\pop}(\type^{\pop}) > 0$ then, 
\begin{equation}
\label{eq:cost+alpha}  
\Util_{\act}^{\pop,\type^{\pop}}(\widehat\flowprof)
\ge
\Util_{\actalt}^{\pop,\type^{\pop}}(\widehat\flowprof)-\alpha\prob(\type^{\pop})
\ge
 \Util_{\actalt}^{\pop,\type^{\pop}}(\widehat\flowprof)-\alpha,
\end{equation}
where $\Util_{\act}^{\pop,\type^{\pop}}$ is defined as in \eqref{eq:cost-ANF}.
This implies that, for all $\pop\in\pops$, for all $\type^{\pop}\in\types^{\pop}$,
\begin{equation}
\label{eq:cost+alpha-2}  
\sum_{\act\in\actions}\frac{\flow_{\act}^{\pop}(\type^{\pop})}{\sizepop^{\pop}}\Util_{\act}^{\pop,\type^{\pop}}(\widehat\flowprof)
\ge
\max_{\actalt\in\actions}\Util_{\actalt}^{\pop,\type^{\pop}}(\widehat\flowprof)-\alpha
=
\max_{\widehat\flowprofz}\sum_{\act\in\actions}\frac{\flowz_{\act}^{\pop}(\type^{\pop})}{\sizepop^{\pop}}\Util_{\act}^{\pop,\type^{\pop}}(\widehat\flowprof)-\alpha.
\end{equation}
That is, for all $\pop\in\pops$, for all $\type^{\pop}\in\types^{\pop}$ and all $\widehat\flowprofz$,
\begin{equation}
\label{eq:gamma-alpha}  
\sum_{\act\in\actions}\parens*{\flow_{\act}^{\pop}(\type^{\pop})-\flowz_{\act}^{\pop}(\type^{\pop})} \Util_{\act}^{\pop,\type^{\pop}}(\widehat\flowprof)
\ge
-\sizepop^{\pop}\alpha.
\end{equation}
Summing up over $\pop$ and $\type^{\pop}$, we get,
\begin{equation}
\label{eq:inner-prod}  
\inner*{\nabla\potential_{\prtype}( \widehat\flowprof)}{\widehat\flowprof-\widehat\flowprofz} =\sum_\pop \sum_{\type^{\pop}} \sum_{\act\in\actions}\parens*{\flow_{\act}^{\pop}(\type^{\pop})-\flowz_{\act}^{\pop}(\type^{\pop})} \Util_{\act}^{\pop,\type^{\pop}}(\widehat\flowprof)
\ge
-\sum_\pop\sizepop^{\pop}\alpha=-\alpha,
\end{equation}
and since $\potential_{\prtype}$ is concave, this gives
\begin{equation}
\label{eq:less-alpha}  
\potential_{\prtype}(\widehat\flowprof) - \potential_{\prtype}(\widehat\flowprofz) 
\ge
\inner*{\nabla\potential_{\prtype}( \widehat\flowprof)}{\widehat\flowprof-\widehat\flowprofz}
\ge
-\alpha.
\end{equation}
Thus $\potential_{\prtype}(\widehat\flowprof)\ge\max\potential_{\prtype} - \alpha$. 
That is, any \acl{BWalpE} is an $\alpha$-maximizer of the potential.
\end{proof}

From strict concavity, let $(\flowprof^{*}(\typeprof))_{\typeprof}$ be the  unique \acl{BWE} flow vector. 
We argue now that for any $\alpha$-maximizer $\widehat\flowprof$ of  $\potential_{\prtype}$, the flow vector $(\flowprof(\typeprof))_{\typeprof}$ is close to $(\flowprof^{*}(\typeprof))_{\typeprof}$.

\begin{claim}
\label{cl:epsMin}
For all $\delta>0$, there exists $\alpha>0$ such that for all $\widehat\flowprof$, $\potential_{\prtype}(\widehat\flowprof)\geq\max\potential_{\prtype}-\alpha$ implies that for all $\act$ and all $\typeprof$, 
$\abs*{\flow_{\act}(\typeprof)-\flow^{*}_{\act}(\typeprof)} \leq \delta$. \end{claim} 
 
\begin{proof}[Proof of \cref{cl:epsMin}]
By contradiction, suppose that there exists $\delta>0$ such that for all $\alpha=1/\run$ with $\run\in\naturals$, there exists $\widehat\flowprof_{\run}$ with $\potential_{\prtype}(\widehat\flowprof_{\run})\geq\max\potential_{\prtype} - 1/\run$ and  $\max_{\act,\typeprof} \abs*{\flow_{n,\act}(\typeprof)-\flow^{*}_{\act}(\typeprof)} > \delta$. 
By compactness, take a subsequence and assume that $\widehat\flowprof_{\run}\to\widehat\flowprof_\infty$ as $\run\to\infty$.
Letting $\run\to\infty$, we get that $\potential_{\prtype}(\widehat\flowprof_\infty)\geq\max\potential_{\prtype}$ and $\max_{\act,\typeprof} \abs*{\widehat\flow_{\infty,\act}(\typeprof)-\flow^{*}_{\act}(\typeprof)} \ge \delta$. 
This contradicts the uniqueness of $\flow^{*}_{\act}(\typeprof)$.
\end{proof}
 
\medskip
Take a \acl{BWalpE}  with flow vector $(\widetilde\flowprof(\typeprof))_{\typeprof}$ and  outcome $\widetilde\outcome$. We have for each $\stateth\in\states$,
\begin{equation*}
\begin{split}
\abs*{\int  \descost(\flowprof,\stateth) \diff \eq{\outcome}(\flowprof\mid \stateth)
-
\int  \descost(\flowprof,\stateth) \diff \widetilde\outcome(\flowprof\mid \stateth)}
&=
\abs*{\sum_{\typeprof}\prtype(\typeprof\mid\stateth)   \descost(\flowprof^{*}(\typeprof),\stateth) 
-
\sum_{\typeprof}\prtype(\typeprof\mid\stateth)   \descost(\widetilde\flowprof(\typeprof),\stateth)}\\
&\leq
\sum_{\typeprof}\prtype(\typeprof\mid\stateth) \abs*{\descost(\flowprof^{*}(\typeprof),\stateth) - \descost(\widetilde\flowprof(\typeprof),\stateth)}\\
&\leq
\modcont_{\descost}(\delta).
\end{split}
\end{equation*}
Summing up, for all $\stateth\in\states$, we have the following inequalities.
By construction, we have
\begin{equation}
\label{eq:construction} 
\abs*{\int  \descost(\flowprof,\stateth) \diff \outcome(\flowprof\mid \stateth)- \int  \descost(\flowprof,\stateth) \diff \breve\outcome_{\infostr}(\flowprof\mid \stateth)} \leq \alpha.
\end{equation}
Moreover, because the obedient  interim profile is a \acl{BWalpE}, we have
\begin{equation}
\label{eq:obedient-BWE}  
\abs*{\int  \descost(\flowprof,\stateth) \diff  \breve\outcome_{\infostr}(\flowprof\mid \stateth)- \int  \descost(\flowprof,\stateth) \diff \eq{\outcome}(\flowprof\mid \stateth)} \leq \modcont_{\descost}(\delta)
\end{equation}
and for any  \acl{BWalpE}:
\begin{equation}
\label{eq:forall-BWE} 
\abs*{\int  \descost(\flowprof,\stateth) \diff \eq{\outcome}(\flowprof\mid \stateth)-\int  \descost(\flowprof,\stateth) \diff \widetilde\outcome(\flowprof\mid \stateth)} \leq  \modcont_{\descost}(\delta).
\end{equation}
It follows that 
\begin{equation}
\label{eq:less-apha+2omega} 
\abs*{\int  \descost(\flowprof,\stateth) \diff \outcome(\flowprof\mid \stateth)-\int  \descost(\flowprof,\stateth) \diff \widetilde\outcome(\flowprof\mid \stateth)} \leq \alpha+2  \modcont_{\descost}(\delta),
\end{equation}
which is smaller than $\varepsilon$ for $\alpha,\delta$ small enough.
\end{proof}

\begin{proof}
[Proof of \cref{thm:fullBI}]    
The result follows from the fact that  the unique equilibrium outcome $\eq{\outcome}$ is an \acl{BWalpE} outcome for any $\alpha>0$. 
As in the proof of \cref{thm:partialBI}, if  $\outcome$ has finite support with  rational flows, then we can choose $\varepsilon = 0$, and the obedient outcome  $\breve\outcome_{\infostr}$ is the unique  \acl{BWE} outcome $\eq{\outcome}$.
\end{proof}

The next theorem shows that in games with a weakly concave potential we can fully implement the ex-ante expected total payoff of any  \acl{BCWE}, where  the total payoff in state $\stateth$ given flows $\flowprof$ is:
\begin{equation}
\label{eq:SC} 
\socialcost(\flowprof,\stateth) \coloneqq \sum_{\act}\flow_{\act}\util_{\act}(\flowprof,\stateth).
\end{equation}

\begin{theorem}
\label{thm:fullBI2}
Let $\game$ be a basic game with a concave potential in each state and let  $\outcome$ be a \ac{BCWE}.  
For all
$\varepsilon>0$,  there exists a direct information structure $\infostr$ such that for each \acl{BWE} outcome $\eq{\outcome}$ of $(\game,\infostr)$,
\begin{equation}
\label{eq:diff-mu-mu-star-2}   
\abs*{\sum_\stateth\prior(\stateth)\int  \socialcost(\flowprof,\stateth) \diff \outcome(\flowprof\mid \stateth)
-
\sum_\stateth\prior(\stateth)\int  \socialcost(\flowprof,\stateth) \diff \eq{\outcome}(\flowprof\mid \stateth)} \leq \varepsilon.
\end{equation}
\end{theorem}

As before,  if  $\outcome$ has finite support and  rational flows,  we can choose $\varepsilon = 0$ and we know that the obedient outcome  $\breve\outcome_{\infostr}$ is a \acl{BWE} outcome.

\cref{thm:fullBI2} applies to  Cournot games. 
In \cref{ex:cournot-binary}, it implies that there exists an information structure (the one described on page~\pageref{page:IS-Cournot}) that implements a \ac{BCWE}  which maximizes total profit in \emph{every} \acl{BWE} of the basic game extended with this information structure. 
However, \cref{thm:fullBI2} does not apply to \cref{ex-El-Farol}, which is a potential game but lacks concavity. 
In this case, \cref{thm:partialBI} implies that there exists an information structure (the one described on page~\pageref{page:IS-elfarol}) that implements the \ac{BCWE} which maximizes the total payoff, in a \acl{BWE} of the basic game extended with this information structure, but other, less efficient equilibria also exist.%
\footnote{\cref{thm:fullBI2} does not apply to games with positive payoff externalities such as those in \citet{LiSonZha:JET2023}.
However, \citet{LiSonZha:JET2023} provide an example showing that full implementation may still be possible with complex information structures involving an infinite set of signals.}

\begin{proof} 
[Proof of \cref{thm:fullBI2}]
Assume that $\outcome$ is a \ac{BCWE} with finite support, choose the function $\descost$ as the social payoff function $\socialcost$,  and  fix $\alpha>0$. Consider  the information structure $\infostr$ constructed in the  proof of  \cref{thm:partialBI}
with parameter $\error>0$, such that the obedient flow profile $\breve\flowprof_{\infostr}$ with outcome $\breve\outcome_{\infostr}$ is a \acl{BWalpE} and 
\begin{equation}
\label{eq:diff-alpha-2}    
\abs*{\int  \socialcost(\flowprof,\stateth) \diff \outcome(\flowprof\mid \stateth)
-
\int  \socialcost(\flowprof,\stateth) \diff \breve\outcome_{\infostr}(\flowprof\mid \stateth)} 
\leq \alpha, \quad\forall \stateth\in\states.
\end{equation}
We know that the  set of Wardrop equilibria  of the associated multi-population game with complete information  is the set of  maximizers of the potential  \citep[Proposition~3.1]{San:JET2001}.  
Because there is a concave potential in each state, the potential defined in \cref{eq:potential-theta} is a concave function and the set of Wardrop equilibria is convex.
Also, from  \citet[Proposition 4]{KoeScaTom:MOR2025}, all Wardrop equilibria have the same payoff profile. That is, for any Wardrop equilibrium flow vectors $\flowprof^{*}=(\flowprof^{*}(\typeprof))_{\typeprof}$ and $\flowprof^{**}=(\flowprof^{**}(\typeprof))_{\typeprof}$, for every $\pop,\type^\pop$,
\begin{equation}
\label{eq:Costa*a**} 
\Util_{\act^{*}}^{\pop,\type^{\pop}}(\flowprof^{*})= \Util_{\act^{**}}^{\pop,\type^{\pop}}(\flowprof^{**}),
\end{equation}
where $\act^{*}$ (resp. $\act^{**}$) is any action in the support of $\flowprof^{*}$ (resp.  $\flowprof^{**}$) for population $(\pop,\type^{\pop})$. 
It follows that all Wardrop equilibria have the same total payoff, where the total payoff of an  interim flow profile $\widehat\flowprof$ is,
\begin{equation}
\label{eq:TC}  
\widehat\socialcost(\widehat\flowprof):= \sum_{\pop,\type^\pop}\sum_{\act}\flow_{\act}^{\pop}(\type^{\pop})\Util_{\act}^{\pop,\type^{\pop}}(\widehat\flowprof).
\end{equation}

From \eqref{eq:cost-ANF}, the total payoff in the population game is equal to the expected total payoff in the basic game, i.e.,
\begin{equation}
\label{eq:TC-y-hat}
\begin{split}\widehat\socialcost(\widehat\flowprof)
&=
\sum_{\pop,\type^\pop}\sum_{\act}
\flow_{\act}^{\pop}(\type^{\pop})
 \sum_{\stateth,\typeprof^{-\pop}}
\prior(\stateth) \prtype\parens*{\type^{\pop},\typeprof^{-\pop}\mid\stateth} \util_{\act}(\flowprof(\typeprof),\stateth)\\
&=
\sum_{\stateth}\prior(\stateth)
\sum_{\typeprof} \prtype\parens*{\typeprof\mid\stateth}
\sum_{\act}
\sum_{\pop}\flow_{\act}^{\pop}(\type^{\pop})
 \util_{\act}(\flowprof(\typeprof),\stateth)\\
&=
\sum_{\stateth}\prior(\stateth)
\sum_{\typeprof} \prtype\parens*{\typeprof\mid\stateth}
\sum_{\act}
\flow_{\act}(\typeprof)
 \util_{\act}(\flowprof(\typeprof),\stateth)\\
&=
\sum_{\stateth}\prior(\stateth)
\sum_{\typeprof} \prtype\parens*{\typeprof\mid\stateth}
\socialcost(\flowprof,\stateth).
\end{split}
\end{equation}
It follows that the expected total payoff is the same in all \acl{BWE} outcomes of $\game$ with information structure $\infostr$.
From \cref{cl:alpMin}, any \acl{BWalpE} of the game with information structure $\infostr$, is a $\alpha$-maximizer of the associated potential. 
To complete the proof, we use the following argument: 

\begin{claim}
\label{cl:epsMinweak}
For all $\delta>0$, there exists $\alpha>0$ such that for all $\widehat\flowprof$, $\potential_{\prtype}(\widehat\flowprof)\geq\max\potential_{\prtype}-\alpha$ implies 
\begin{equation}
\label{eq:diff-TC} 
\abs*{\widehat\socialcost(\widehat\flowprof)-\widehat\socialcost(\widehat\flowprof^{*})} \leq \delta,
\end{equation}
where $\widehat\flowprof^{*}$ is any  \acl{BWE}  of $\game$ with information structure $\infostr$.
\end{claim} 

\begin{proof}[Proof of \cref{cl:epsMinweak}]
By contradiction, suppose that there exists $\delta>0$ such that for all $\alpha=1/\run$ with $\run\in\naturals$, there exists $\widehat\flowprof_{\run}$ with $\potential_{\prtype}(\widehat\flowprof_{\run})\geq\max\potential_{\prtype} - 1/\run$ and  $\abs*{\widehat\socialcost(\widehat\flowprof_{\run})-\widehat\socialcost(\widehat\flowprof^{*})} > \delta$. 
By compactness, take a subsequence and assume that $\widehat\flowprof_{\run}\to\widehat\flowprof_\infty$ as $\run\to\infty$.
Letting $\run\to\infty$, we get that $\potential_{\prtype}(\widehat\flowprof_\infty)\geq\max\potential_{\prtype}$ and $\abs*{\widehat\socialcost(\widehat\flowprof_\infty)-\widehat\socialcost(\widehat\flowprof^{*})} \ge \delta$. 
This is a contradiction because   $\widehat\flowprof_\infty$ would then be a maximizer of the potential, thus a \acl{BWE}, with a different total payoff.
\end{proof}
 
This concludes the proof of \cref{thm:fullBI2}.
\end{proof}

%
%
%
%

\section{Optimality of Deterministic \ac{BCWE}}
\label{se:optimality-deterministic}

The following proposition identifies a class of games with binary action sets in which, in order to maximize the designer's payoff function $\descost$,  attention can be restricted to deterministic flows conditional on the state.  
 A \ac{BCWE} $\outcome$ is said to be \emph{deterministic} if, in each state $\stateth\in\states$, a single flow profile is implemented with probability one. 
 Hence, a deterministic \ac{BCWE} can be identified with a mapping $\flowprof(\argdot)\colon \states\to\flows$ such that, for all $\act,\actalt\in\actions$,
 \begin{equation}
 \label{eq:BDWE-def}
 \sum_{\stateth\in\states} \prior(\stateth)\,
 \flow_{\act}(\stateth)\util_{\act}(\flowprof(\stateth),\stateth)
 \ge
 \sum_{\stateth\in\states} \prior(\stateth)\,
 \flow_{\act}(\stateth)\util_{\actalt}(\flowprof(\stateth),\stateth).
 \end{equation}
 Under complete information, the notion of deterministic \ac{BCWE} coincides with that of a \acl{WE}.

\begin{proposition}
\label{pr:BDWE-cost}
Consider a game where $\actions$ contains exactly two actions; 
for every $\act\in\actions$ and $\stateth\in\states$, the mapping $\flowprof\mapsto\util_{\act}(\flowprof,\stateth)$ is convex and the mapping $\flowprof\mapsto \flow_{\act}\util_{\act}(\flowprof,\stateth)$ is concave. 
Moreover, assume that the designer payoff  function $\flowprof\mapsto\descost(\flowprof,\stateth)$ is concave for every $\stateth\in\states$. 
Then, for every \ac{BCWE} of $\game$, there exists a deterministic \ac{BCWE} for which the expected designer payoff is weakly higher.
\end{proposition}

\begin{proof}
Consider a \ac{BCWE} outcome $\outcome$. 
Then
\begin{equation}
\begin{split}
&\forall \act,\actalt\in\actions,\ 
\sum_{\stateth\in\states}\prior(\stateth)\int \flow_{\act}\util_{\act}(\flowprof,\stateth)\diff\outcome(\flowprof\mid\stateth)
\ge
\sum_{\stateth\in\states}\prior(\stateth)\int \flow_{\act}\util_{\actalt}(\flowprof,\stateth)\diff\outcome(\flowprof\mid\stateth).
\end{split}
\end{equation}
Since there are two actions, we have $\flow_{\actalt}=1-\flow_{\act}$ for $\actalt\neq\act$. 
The incentive constraint corresponding to recommendation $\act$ and deviation $\actalt\neq\act$ can therefore be rewritten as
\begin{equation}
\label{eq:BCWE-a}
\sum_{\stateth\in\states}\prior(\stateth)\int \flow_{\act}\util_{\act}(\flowprof,\stateth)\diff\outcome(\flowprof\mid\stateth)
\ge
\sum_{\stateth\in\states}\prior(\stateth)\int (1-\flow_{\actalt})\util_{\actalt}(\flowprof,\stateth)\diff\outcome(\flowprof\mid\stateth),
\end{equation}
or equivalently,
\begin{equation}
\label{eq:BCWE-b}
\begin{split}
&\sum_{\stateth\in\states}\prior(\stateth)\int \flow_{\act}\util_{\act}(\flowprof,\stateth)\diff\outcome(\flowprof\mid\stateth)
+\sum_{\stateth\in\states}\prior(\stateth)\int \flow_{\actalt}\util_{\actalt}(\flowprof,\stateth)\diff\outcome(\flowprof\mid\stateth)\\
&\quad\ge
\sum_{\stateth\in\states}\prior(\stateth)\int \util_{\actalt}(\flowprof,\stateth)\diff\outcome(\flowprof\mid\stateth).
\end{split}
\end{equation}

For each state $\stateth$, define the average flow profile
\begin{equation}
\label{eq:average-flow-profile}    
\bar\flowprof(\stateth)\coloneqq \int \flowprof\diff\outcome(\flowprof\mid\stateth),
\quad
\text{with components }
\bar\flow_{\act}(\stateth)=\int \flow_{\act}\diff\outcome(\flowprof\mid\stateth).
\end{equation}
By concavity of $\flowprof\mapsto \flow_{\act}\util_{\act}(\flowprof,\stateth)$ and convexity of $\flowprof\mapsto \util_{\actalt}(\flowprof,\stateth)$, Jensen’s inequality yields
\begin{align}
\label{eq:Jensen-1}
\sum_{\stateth\in\states}\prior(\stateth)\int \flow_{\act}\util_{\act}(\flowprof,\stateth)\diff\outcome(\flowprof\mid\stateth)
&\le
\sum_{\stateth\in\states}\prior(\stateth)\,
\bar\flow_{\act}(\stateth)\util_{\act}(\bar\flowprof(\stateth),\stateth),\\
\label{eq:Jensen-1'}
\sum_{\stateth\in\states}\prior(\stateth)\int \flow_{\actalt}\util_{\actalt}(\flowprof,\stateth)\diff\outcome(\flowprof\mid\stateth)
&\le
\sum_{\stateth\in\states}\prior(\stateth)\,
\bar\flow_{\actalt}(\stateth)\util_{\actalt}(\bar\flowprof(\stateth),\stateth),\\
\label{eq:Jensen-2}
\sum_{\stateth\in\states}\prior(\stateth)\int \util_{\actalt}(\flowprof,\stateth)\diff\outcome(\flowprof\mid\stateth)
&\ge
\sum_{\stateth\in\states}\prior(\stateth)\,
\util_{\actalt}(\bar\flowprof(\stateth),\stateth).
\end{align}
Thus, \cref{eq:BCWE-b} implies that $\stateth\mapsto\bar\flowprof(\stateth)$ satisfies the deterministic \ac{BCWE} incentive constraints. 
Finally, concavity of $\flowprof\mapsto\descost(\flowprof,\stateth)$ implies
\begin{equation}
\label{eq:concavity-y}  
\sum_{\stateth\in\states}\prior(\stateth)\int \descost(\flowprof,\stateth)\diff\outcome(\flowprof\mid\stateth)
\le
\sum_{\stateth\in\states}\prior(\stateth)\descost(\bar\flowprof(\stateth),\stateth),
\end{equation}
which completes the proof.
\end{proof}

\appendix
%
%
%
%

\newpage

\section{List of Symbols and Acronyms}
\label{se:symbols}

\begin{longtable}{p{.15\textwidth} p{.82\textwidth}}

$\actions$ &  action set, defined in \cref{de:BANG}\\
$\act$ & action\\
$\actalt$ & action\\
$\BWespE(\infostr)$ & set of \acl{BWepsE} outcomes of   $\game$ with information structure~$\infostr$\\
$\cost_{\edge}$ & cost of resource $\edge$\\
$\smooth^{h}$ & class of functions with a continuous $h$-th derivative\\
$\edge$ & resource in a congestion game\\
$\edges$ & resource set in a congestion game\\
$\infostr$ & $\infostr=\parens*{\sizepopprof, \types, \prtype}$, information structure, defined in \cref{de:info-structure}\\
$\pop$ & population\\
$\pops$ & population set\\
$\lipschitz$ & Lipschitz constant\\
$\prior$ & probability distribution on $\states$, defined in \cref{de:BANG}\\
$\prob(\type^{\pop})$ & $\sum_{\stateth\in\states}\sum_{\typeprof^{-\pop}\in\types^{-\pop}}  \prior(\stateth) \prtype(\type^{\pop},\typeprof^{-\pop} \mid \stateth)$, total probability of type $\type^{\pop}$\\

$\Qmeas^{\pop}(\act)$ & $\sum_{\stateth\in\states}\sum_{\flowprof\in\suppflows}\prior(\stateth)\outcome(\flowprof\mid\stateth)\flow_{\act}^{\pop}$, defined in \eqref{eq:Q-k}\\
$\quantity(\flowprof)$ & aggregate quantity in an aggregative game, defined before \eqref{eq:potential-congestion5}\\
$\types$ & $\types = \times_{\pop\in\pops} \types^{\pop}$, defined in \cref{de:info-structure}\\
$\types^{\pop}$ & set of types for population $\pop$, defined in \cref{de:info-structure}\\
$\types^{-\pop}$ & $\times_{\popalt\neq \pop} \types^{\popalt}$\\
$\util_{\act}$ & payoff of action $\act$, defined in \cref{de:BANG}\\
$\util_{\act}^{\pop,\type^{\pop}}$ & payoff functions of population  $(\pop,\type^{\pop})$, defined in \eqref{eq:caktk}\\
$\utilprof$ & $\parens{\util_{\act}}_{\act\in\actions}$, defined in \cref{de:BANG}\\
$(\util_{\act}^{\pop,\type^\pop}(\widehat\flowprof))_{\pop,\type^\pop,\act}$ & interim payoff profile\\
$\Util_{\act}^{\pop,\type^{\pop}}$ & payoff function of population  $(\pop,\type^{\pop})$, defined in \eqref{eq:cost-ANF}\\
$\socialcost(\flowprof,\stateth)$ & $\sum_{\act}\flow_{\act}\util_{\act}(\flowprof,\stateth)$, defined in \eqref{eq:SC}\\
$\WE(\game)$ & \phantom{}\!\aclp{WE}  of $\game$\\
$\load_{\edge}$ & load on resource $\edge$, defined after \eqref{eq:congestion}\\
$\flowaktype$ & flow on action $\act$ of population $\pop$, when type $\type^{\pop}$ is observed\\ 
$\flow_{\act}(\typeprof)$ & $\sum_{\pop\in\pops} \flowaktype$, interim total flow on action $\act$\\
$\breve\flowprof_{\infostr, \act}^{\pop}(\type^{\pop})$ & $\sizepop^{\pop}\cdot\ind\braces*{\type^{\pop}=\act}$\\
$\flowprof$ & flow profile\\
$\eq{\flowprof}$ & equilibrium flow profile\\
$\flowktypeprof$ & $\parens*{\flowaktype}_{\act\in\actions}\in \simplex_{\sizepop^{\pop}}(\actions)$, interim flow profile of population $\pop$ given type $\type^{\pop}$\\
$\flowprof(\typeprof)$ & $\parens*{\flow_{\act}(\typeprof)}_{\act\in\actions}\in \simplex(\actions)$,  interim total flow profile given $\typeprof$\\
$\widehat{\flowprof}$ & $\parens*{\flowktypeprof}_{\pop,\type^{\pop}} 
\in \bigtimes_{\pop\in\pops}\bracks*{\bigtimes_{\type^{\pop}
\in\types^{\pop}}\simplex_{\sizepop^{\pop}}(\actions)}$, interim flow profile, defined in \eqref{eq:IFP}\\
$\breve\flowprof_{\infostr}$ &  obedient interim flow profile\\
$\flows$ & $\simplex_{\sizepop}(\actions)$, set of flow profiles, defined before \cref{de:BANG}\\
$\sizepop^{\pop}$ & size of population $\pop$, defined in \cref{de:info-structure}\\
$\sizepopprof$ & vector of population sizes, defined in \cref{de:info-structure}\\
$\game$ & $\parens*{\actions, \states,\prior,\utilprof}$, basic game, defined in \cref{de:BANG}  \\
$\dirac_{\act}$ & Dirac measure on $\act$\\
$\simplex(\setA)$ & $\simplex_{1}(\setA)$, simplex of probability measures on $\setA$\\
$\simplex_{\sizepop}(\setA)$ & $\braces*{\flowprof\in\reals^{\setA} \colon \forall \act\in \setA, \flow_{\act} \geq 0, \sum_{\act\in\setA} \flow_{\act}  = \sizepop}$, defined in \eqref{eq:simplex}\\
$\stateth$ & state\\
$\states$ & state space, defined in \cref{de:BANG}\\
$\outcome$ & outcome of the game, defined after \cref{de:BANG}\\
$\eq{\outcome}$ & \acl{BWE} outcome \\
$\breve\outcome_{\infostr}$ & outcome associated to the obedient interim flow profile\\
$\prtype$ & mapping from $\states$ to $\simplex(\types)$, defined in \cref{de:info-structure}\\
$\type^{\pop}$ & possible type for population $\pop$\\
$\typeprof$ & $(\type^{\pop})_{\pop\in\pops}$ type profile\\
$\potential_{\stateth}$ & potential, defined in \cref{de:potential}\\
$\descost$ & designer welfare function, first used in \cref{thm:partialBI}\\
$\modcont(\argdot)$ & modulus of continuity\\
$\ind$ & indicator function\\

\ac{BCWE} & \acl{BCWE}, defined in \cref{de:BCWE}\\
\ac{BWE} & \acl{BWE}, defined after \eqref{eq:epsBWE-cond}\\
\ac{BWepsE} & \acl{BWepsE}, defined in \cref{de:BWepsE}\\
\end{longtable}

\bibliographystyle{apalike}
\bibliography{bibIDLG.bib}

\end{document}